\documentclass[11pt]{article}

\usepackage[a4paper,margin=1in]{geometry}
\usepackage{amsmath,amssymb,amsthm,mathtools}
\usepackage{bm}
\usepackage{braket}
\usepackage{graphicx}
\usepackage{xcolor}
\usepackage{enumitem}
\usepackage{float}
\usepackage{eso-pic}
\usepackage{aliascnt}
\usepackage{hyperref}
\usepackage[nameinlink,capitalize]{cleveref}

\hypersetup{
  colorlinks=true,
  linkcolor=blue,
  citecolor=blue,
  urlcolor=blue,
  pdftitle={Unweighted Gapped Clique Homology is QMA1(G2)-complete},
  pdfauthor={Ryu Hayakawa}
}

\graphicspath{{figures/}}

\newtheorem{theorem}{Theorem}[section]

\newaliascnt{lemma}{theorem}
\newtheorem{lemma}[lemma]{Lemma}
\aliascntresetthe{lemma}

\newaliascnt{proposition}{theorem}
\newtheorem{proposition}[proposition]{Proposition}
\aliascntresetthe{proposition}

\newaliascnt{corollary}{theorem}
\newtheorem{corollary}[corollary]{Corollary}
\aliascntresetthe{corollary}

\theoremstyle{definition}
\newaliascnt{definition}{theorem}
\newtheorem{definition}[definition]{Definition}
\aliascntresetthe{definition}

\newaliascnt{remark}{theorem}

\aliascntresetthe{remark}

\crefname{theorem}{Theorem}{Theorems}
\Crefname{theorem}{Theorem}{Theorems}
\crefname{lemma}{Lemma}{Lemmas}
\Crefname{lemma}{Lemma}{Lemmas}
\crefname{proposition}{Proposition}{Propositions}
\Crefname{proposition}{Proposition}{Propositions}
\crefname{corollary}{Corollary}{Corollaries}
\Crefname{corollary}{Corollary}{Corollaries}
\crefname{definition}{Definition}{Definitions}
\Crefname{definition}{Definition}{Definitions}
\crefname{remark}{Remark}{Remarks}
\Crefname{remark}{Remark}{Remarks}

\newcommand{\C}{\mathbb C}
\newcommand{\R}{\mathbb R}
\newcommand{\Z}{\mathbb Z}
\newcommand{\poly}{\mathrm{poly}}
\newcommand{\Cl}{\mathrm{Cl}}
\newcommand{\im}{\mathrm{im}}

\newcommand{\level}{\ell}
\newcommand{\Enc}{\operatorname{Enc}}
\newcommand{\hEnc}{\widehat{\mathbf{Enc}}}
\newcommand{\regcycle}[1]{\psi_{#1}}
\newcommand{\hregcycle}[1]{\widehat{\bm{\psi}_{#1}}}
\newcommand{\hregvertex}[1]{\ket{\widehat{\bm{v}}_{#1}}}
\newcommand{\mloc}{m_{\mathrm{loc}}}

\newcommand{\QMAone}{\mathsf{QMA}_1}
\newcommand{\SharpP}{\mathsf{\#P}}
\newcommand{\PSPACE}{\mathsf{PSPACE}}
\newcommand{\Gtwo}{\mathcal{G}_2}

\newcommand{\DeltaAug}{\Delta^{\mathrm{aug}}}
\newcommand{\bd}{\widehat{\bm d}}
\newcommand{\bpartial}{\widehat{\bm{\partial}}}
\newcommand{\bDelta}{\widehat{\bm{\Delta}}}
\newcommand{\bT}{\widehat{\bm T}}
\newcommand{\bfmult}{\widehat{\bm f}}
\newcommand{\wtilde}{\widetilde w}
\newcommand{\BlockSym}{\mathfrak S}
\newcommand{\GQ}{G_Q}
\newcommand{\bGQ}{\widehat{\bm{G}}_Q}
\newcommand{\bR}{\widehat{\bm{R}}}
\newcommand{\GH}{G_H}
\newcommand{\bGH}{\widehat{\bm{G}}_H}
\newcommand{\GammaH}{\Gamma_H}
\newcommand{\bGamma}{\widehat{\bm{\Gamma}}}
\newcommand{\bGammaH}{\widehat{\bm{\Gamma}}_H}
\newcommand{\bGcal}{\widehat{\bm{\mathcal G}}}

\newcommand{\ratio}{\rho}
\newcommand{\pgap}{p}
\newcommand{\pH}{p_H}
\newcommand{\KK}{King--Kohler}

\title{Unweighted Gapped Clique Homology is
\texorpdfstring{$\mathsf{QMA}_1$}{QMA1}-complete}
\author{Ryu Hayakawa\thanks{The Hakubi Center and
Yukawa Institute for Theoretical Physics, Kyoto University.\\
Email: \texttt{ryu.hayakawa@yukawa.kyoto-u.ac.jp}}}
\date{\today}

\begin{document}
\AddToShipoutPictureFG*{%
  \AtPageUpperLeft{%
    \raisebox{-0.65in}{\makebox[\paperwidth][r]{\small YITP-26-99\hspace{1in}}}%
  }%
}
\maketitle

\begin{abstract}
    Deciding whether the clique complex of a given graph has nontrivial homology in a given dimension, under vertex-product weighting and an inverse-polynomial spectral gap promise on the combinatorial Hodge Laplacian, is known to be \(\mathsf{QMA}_1\)-hard and contained in \(\mathsf{QMA}\) by King and Kohler (FOCS 2024). The vertex weights are essential in the known proof, where they provide the scale separation needed for the spectral gap analysis. We prove that, when every vertex has weight one, the problem remains \(\mathsf{QMA}_1^{\mathcal{G}_2}\)-hard and is contained in \(\mathsf{QMA}_1^{\mathcal{G}_2}\), where \(\mathsf{QMA}_1^{\mathcal{G}_2}\) is \(\mathsf{QMA}_1\) with the universal gate set \(\mathcal{G}_2=\{\mathsf{X},\mathsf{CX},\mathsf{CCX},H\otimes H\}\). The construction replaces weight by expansion: each vertex of the weighted complex is blown up into a clique whose size encodes its weight. The block sizes are chosen so that symmetric averages reproduce the weighted Hodge metric. The symmetric sector therefore carries the weighted Laplacian up to a common scalar factor, while a local averaging argument gives a uniform lower bound on the orthogonal complement to the symmetric sector. Hence the weighted gap analysis transfers to an unweighted clique complex without introducing either additional low-energy states or spurious homology. Containment in \(\mathsf{QMA}_1^{\mathcal{G}_2}\) follows from Rudolph's exact linear combination of unitaries simulation of sparse integer clique Laplacians. The result shows that the gap promise, rather than vertex weighting, is the source of the complexity of gapped clique homology.
\end{abstract}


\section{Introduction}
\label{sec:introduction}

Computing homology of clique complexes is a basic primitive in topological data
analysis and a natural target for quantum algorithms. In the unweighted
setting, the \(k\)-th combinatorial Hodge Laplacian
\[
  \Delta_k=\partial_{k+1}\partial_{k+1}^{\dagger}
  +\partial_k^{\dagger}\partial_k
\]
has kernel canonically isomorphic to the \(k\)-homology group \(H_k\). Thus a
spectral promise on \(\Delta_k\) gives a gapped version of the homology
decision problem.

The complexity of homology depends strongly on how the complex is specified.
For an explicitly listed finite complex, homology is efficiently computable by
linear algebra. For succinct descriptions, the situation changes: Adamaszek and
Stacho proved \(\mathsf{NP}\)-hardness for homology of clique and independence
complexes of graphs, and identified a co-chordal regime where nontrivial
homology has a combinatorial witness through cross-cycles, proving
\(\mathsf{NP}\)-completeness for clique complexes of co-chordal
graphs~\cite{adamaszek2016complexity}. On the quantum-proof side, Crichigno and
Kohler showed that clique homology is
\(\QMAone\)-hard~\cite{crichigno2022clique}. King and
Kohler strengthened this to a
gapped clique-homology problem on weighted graphs~\cite{king2024gapped}.
A complementary result shows that homology with orientable filtrations has an
\(\mathsf{MA}\)-complete regime~\cite{hayakawa2025computationalMA}. For
Betti-number computation in the usual topological-data-analysis input model,
Schmidhuber and Lloyd proved \(\SharpP\)-hardness of exact computation and
\(\mathsf{NP}\)-hardness of multiplicative
approximation~\cite{schmidhuber2024complexitytheoretic}. Rudolph
found an exact-gate setting where the gapped problem is
\(\QMAone\)-complete and the gapless decision problem is
\(\PSPACE\)-complete~\cite{rudolph2025universala}. Rayudu obtained a parallel
weighted-complex hardness result for the combinatorial Laplacian of independence
complexes~\cite{rayudu2025fermionic}.

This proof-system landscape is closely related to quantum algorithms for
topological data analysis. The algorithms of Lloyd, Garnerone, and Zanardi for
Betti-number estimation, and subsequent work on persistent Betti numbers, use
spectral access to combinatorial or persistent Hodge Laplacians under sampling
and gap assumptions~\cite{lloyd2015quantum,hayakawa2022quantum}.
Worst-case evidence for quantum advantage has developed in parallel: harmonic
persistence is \(\mathsf{BQP}_1\)-hard and in \(\mathsf{BQP}\), while
normalized persistence variants are \(\mathsf{DQC1}\)-hard and in
\(\mathsf{BQP}\)~\cite{gyurik2024quantum,lowe2026complexity}. These results
make the status of the gap promise for clique Hodge Laplacians a central
structural question rather than a technical side condition.

Vertex weights play an important role in these results. In the \KK{} gap proof,
gadget vertices receive a small weight \(\lambda\ll1\), which supplies the
scale separation that makes the gadget degrees of freedom decouple. In the
independence-complex result \cite{rayudu2025fermionic}, vertex weights encode
the Hamiltonian coefficients. The unweighted case therefore raises a question:
can gapped hardness survive when every vertex has weight \(1\)?

Our answer is yes: the weighted hardness result can be made fully unweighted.
The key is to replace vertex weights by multiplicities. The construction
encodes the intended qubit space redundantly, without assigning different
weights to vertices, and gives an energy penalty to directions outside this
encoding. In this limited sense it is close in spirit to quantum
error-correcting codes such as the toric code, where information is stored in
homological degrees of freedom. Our main result is the following.

\begin{theorem}[Main theorem, informal]
Unweighted Gapped Clique Homology is \(\QMAone^{\Gtwo}\)-complete, for
\(\Gtwo=\{\mathsf{X},\mathsf{CX},\mathsf{CCX},H\otimes H\}\).
\end{theorem}

\paragraph{Proof overview.}
Let \(\GammaH\) be the weighted gadget complex associated with the source
Hamiltonian originally introduced in \cite{king2024gapped}, and let \(\bGammaH\) be its unweighted blow-up, see \Cref{fig:blowup-correspondence} for an illustrative example. Permuting the
copies inside any vertex block preserves both \(\bGammaH\) and its Hodge
Laplacian. In the target degree, this gives the orthogonal and Laplacian-invariant
decomposition
\[
  C_{2n-1}(\bGammaH;\R)
  =
  \underbrace{\bT_{2n-1}}_{\text{symmetric}}
  \oplus
  \underbrace{\bT_{2n-1}^{\perp}}_{\text{asymmetric}}.
\]
The symmetric sector is spanned by uniform one-copy-per-block averages; the
asymmetric sector consists of chains that retain some information about the
copy labels.
Concretely, the symmetric basis vectors are given by
\[ \big\{ \ket{\widehat{\bm{\sigma}}}
  =
  \ket{\widehat{\bm v_0}}\otimes\cdots \otimes\ket{\widehat{\bm v_{2n-1}}}\big\}_{\sigma=(v_0,\ldots,v_{2n-1}) \in \Gamma_{H,2n-1}},\]
where \(\ket{\widehat{\bm v_i}}\) is the uniform superposition of the blow-up vertices for \(v_i\) with some fixed multiplicities.

On the symmetric sector, \Cref{lem:symmetric-reduction} identifies the
unweighted blow-up Laplacian with the weighted source Laplacian, up to one
global multiplicative factor \(M\), the top block multiplicity:
\[
  \Phi\,
  \bDelta_{2n-1}\big|_{\bT_{2n-1}}
  \,\Phi^{-1}
  =M\Delta_{\GammaH,2n-1}.
\]
Consequently the single- and many-gadget energy analysis transfers to the
unweighted complex with only this common rescaling. In the YES case, the
weight-free filling argument preserves precisely the encoded states in
\(\ker H\); in the NO case, the weighted many-gadget bound supplies a positive
lower bound \(E_{\mathrm{wt}}\) on \(\Delta_{\GammaH,2n-1}\).

The remaining task is to rule out low-energy states that remember the copy
labels. This is the content of \Cref{thm:asymmetric-decoupling}, which gives
the direct estimate
\[
  \bDelta_{2n-1}\big|_{\bT_{2n-1}^{\perp}}
  \succeq
  \bigl(\min_v \bfmult_v\bigr)I.
\]
Consequently the asymmetric sector contains neither spurious homology nor
additional low-energy states. In the NO case the two invariant sectors combine
to give
\[
  \bDelta_{2n-1}
  \succeq
  \min\!\left\{M E_{\mathrm{wt}},\min_v \bfmult_v\right\}I,
\]
where \(\bfmult_v\) is the multiplicity for \(v\) in the blow-up construction.
Choosing integer-valued block multiplicities makes the first term inverse polynomial and
the second uniformly positive. Topologically, the expanded gadgets still kill
exactly the cycles corresponding to violated local states. Spectrally, the
decomposition above prevents the blow-up from introducing additional low-energy
directions. These two facts preserve perfect completeness while transferring
the weighted gap to the unweighted clique complex.

\medskip

This also suggests an outlook for worst-case quantum advantages in topological
data analysis. As discussed in \Cref{sec:discussion}, similar unweighting
questions can be asked for other weighted or filtered homological complexity
results. The present theorem indicates that, at least in this clique-homology
setting, artificial vertex weights are not the source of the hardness.

\paragraph{Remark on the hat and bold notations.}
Throughout the paper, unadorned symbols such as \(\GammaH\) denote weighted base objects while hat-bold
symbols such as \(\bGH\) and \(\bGammaH\) denote their expanded unweighted
counterparts.

\paragraph{Organization.}
\Cref{sec:preliminaries} fixes notation. \Cref{sec:main-results} states the
promise problem and the main theorem. \Cref{sec:blowup-reduction} contains the
blow-up and symmetric reduction. \Cref{sec:expanded-gadget-complex} constructs
the expanded gadget complex. \Cref{sec:energy-analysis} proves the energy
analysis on the symmetric sector and its orthogonal complement.
\Cref{sec:containment} proves
membership in \(\QMAone^{\Gtwo}\). \Cref{sec:discussion} records implications
and open directions. \Cref{app:many-gadget-proof-sketch} gives a proof sketch
for the imported many-gadget estimate.

\section{Preliminaries}
\label{sec:preliminaries}

\subsection{Clique complexes and Hodge Laplacians}

An abstract simplicial complex \(X\) on a finite vertex set \(V\) is a
collection of finite subsets of \(V\) such that if \(\sigma\in X\) and
\(\tau\subseteq\sigma\), then \(\tau\in X\). Its elements are simplices, and
\(\dim\sigma=|\sigma|-1\). This is the combinatorial dimension of a simplex.
When \(\dim\) is applied to a chain, homology, or Hilbert space, it denotes the
ordinary vector-space dimension. For a graph \(G=(V,E)\), its clique complex
\(\Cl(G)\) is the abstract simplicial complex whose simplices are the cliques
of \(G\): a subset \(\sigma\subseteq V\) belongs to \(\Cl(G)\) iff every pair
of distinct vertices in \(\sigma\) is joined by an edge of \(G\). Thus
vertices and edges of \(G\) give the \(0\)- and \(1\)-simplices, while higher
simplices are forced by complete subgraphs.

The real \(k\)-chain space \(C_k(X)\) has a basis indexed by the oriented
\(k\)-simplices. We use the same simplex basis for the corresponding Hilbert
space on which the Hodge Laplacian acts.
The ordinary algebraic boundary
\(\partial_k:C_k(X)\to C_{k-1}(X)\) is defined on an oriented simplex by
summing its oriented codimension-one faces. It satisfies
\(\partial_k\partial_{k+1}=0\). The \(k\)-cycles and \(k\)-boundaries are
\[
  Z_k(X)=\ker \partial_k,
  \qquad
  B_k(X)=\im \partial_{k+1},
\]
and the \(k\)-th homology group is
\[
  H_k(X;\R)=Z_k(X)/B_k(X).
\]
Thus a nonzero homology class is represented by a cycle that is not itself a
boundary. The \(k\)-th Betti number is
\[
  \beta_k(X)=\dim H_k(X;\R).
\]
We also use reduced homology when discussing joins. Its Betti numbers are the
same as the ordinary Betti numbers in positive degrees. In degree \(0\), if
\(X\) is nonempty, reduced homology subtracts the one class represented by a
constant component:
\[
  \widetilde\beta_0(X)=\beta_0(X)-1.
\]

For weighted Hodge operators we use the convention of
\KK{}~\cite[Sec.~7.7]{king2024gapped}. This convention is coboundary-first.
Assign each vertex \(v\) a positive amplitude weight \(w(v)\), put
\[
  w(\sigma)=\prod_{v\in\sigma}w(v),
  \qquad
  \langle \sigma,\tau\rangle_w=w(\sigma)^2\delta_{\sigma,\tau}.
\]
For oriented simplices \(\sigma\subset\tau\), let
\([\tau:\sigma]\in\{+1,-1\}\) be the incidence sign: it is \(+1\) when the
orientation on \(\sigma\) induced from \(\tau\) agrees with the chosen
orientation of \(\sigma\), and \(-1\) otherwise. We set
\([\tau:\sigma]=0\) when \(\sigma\) is not a face of \(\tau\).
Let \(d_k:C_k(X)\to C_{k+1}(X)\) be the ordinary algebraic coboundary,
\[
  d_k\ket{\sigma}
  =
  \sum_{\tau\supset\sigma,\,\dim\tau=k+1}
  [\tau:\sigma]\ket{\tau}.
\]
The weighted boundary is its adjoint
\(\partial_{w,k+1}=(d_k)^\dagger\), and
\[
  \Delta_{w,k}(X)
  =
  \partial_{w,k+1} d_k+d_{k-1}\partial_{w,k}.
\]
Equivalently, in the orthonormal simplex basis
\(\ket{\sigma}'=w(\sigma)^{-1}\ket{\sigma}\), an incidence
\(\tau=\sigma\cup\{v\}\) has matrix coefficient
\[
  \bra{\tau}'d_k\ket{\sigma}'
  =
  [\tau:\sigma]\frac{w(\tau)}{w(\sigma)}
  =
  [\tau:\sigma]w(v).
\]
This convention is important below: the algebraic coboundary is independent of
the weights, whereas its adjoint and the resulting Laplacian depend on them.
In the unweighted case this is the usual chain Hodge Laplacian under the
simplex-basis identification. We therefore continue to write \(\Delta_k\) when
emphasizing the simplex degree. By the Hodge decomposition (see, e.g.,
\cite{goldberg2002combinatorial}),
\(\ker \Delta_k(X) = Z_k(X) \cap
  B_k(X)^\perp\cong H^k(X;\R)\cong H_k(X;\R)\).
Thus the promise problem can be phrased spectrally: the YES case asks for a
zero eigenvalue of \(\Delta_k\), while the NO case promises that this zero
eigenspace is absent and the bottom of the spectrum is bounded below by an
inverse polynomial.

We also use the augmented Hodge Laplacian, mainly when applying join formulas.
Add a degree-\(-1\) basis vector \(\ket{\varnothing}\) for the empty simplex,
write \(C_{-1}(X)=\R\ket{\varnothing}\), and set
\[
  d_{-1}\ket{\varnothing}
  =
  \sum_{v\in V(X)}\ket v.
\]
With \(\partial_0=(d_{-1})^\dagger\), the augmented Laplacian is defined by the
same formula, where \(\partial_{\ell+1}=(d_\ell)^\dagger\) in the unweighted
simplex inner product:
\[
  \DeltaAug_k(X)
  =
  \partial_{k+1}d_k+d_{k-1}\partial_k.
\]
Thus \(\DeltaAug_k=\Delta_k\) for \(k>0\), while
\[
  \DeltaAug_0
  =
  \Delta_0+d_{-1}(d_{-1})^\dagger.
\]
This is the Hodge-Laplacian convention corresponding to reduced homology.

We use real coefficients throughout. The reduction only changes the inner
product on the same underlying real simplex space when it invokes weighted
complexes; the homology groups themselves are unchanged by positive weights.
Weights affect the Hodge Laplacian and its gap, but not the algebraic
coboundary or the underlying (co)homology.

The encoding maps and quantum verifier sometimes use complex amplitudes, so we
also allow complex coefficients in the same simplex basis and write the
resulting space as \(C_k(X;\C)\). Since every Laplacian in the construction is
represented by a real symmetric matrix, allowing complex coefficients leaves
its eigenvalues unchanged and only extends its real kernel by complex linear
combinations. We therefore state the topological conclusions over \(\R\) and
allow complex amplitudes when discussing quantum states.

\subsection{Join calculus}

For simplicial complexes \(X\) and \(Y\) on disjoint vertex sets, their join
\(X*Y\) is the simplicial complex whose simplices are unions
\(\sigma\cup\tau\), with \(\sigma\in X\) and \(\tau\in Y\). For graphs \(G\)
and \(H\), we write \(G*H\) for the graph join: take the disjoint union of
\(G\) and \(H\), and add every edge between a vertex of \(G\) and a vertex of
\(H\). Then
\[
  \Cl(G*H)=\Cl(G)*\Cl(H).
\]
Indeed, a clique in the graph join is exactly the union of a clique in \(G\)
and a clique in \(H\). This compatibility is used throughout the \KK{}
register construction \cite{king2024gapped}.

The join formula is most naturally stated for the augmented Laplacian introduced
above.

\begin{lemma}[Join formula {\cite{horak2011spectra,king2024gapped}}]
\label{lem:horak-jost}
For simplicial complexes \(X,Y\), the chain space of \(X*Y\) decomposes as
\[
  C_k(X*Y)\cong\bigoplus_{i+j=k-1}C_i(X)\otimes C_j(Y),
\]
and on the summand \(C_i(X)\otimes C_j(Y)\),
\[
  \Delta_k^{\mathrm{aug},X*Y}
  =
  \Delta_i^{\mathrm{aug},X}\otimes I
  +I\otimes \Delta_j^{\mathrm{aug},Y}.
\]
The summands with \(i=-1\) or \(j=-1\) represent simplices lying entirely in
the other join factor.
\end{lemma}

Two consequences are used repeatedly. First, harmonic chains in a join are
tensor products of augmented harmonic chains in the factors, in the
appropriate total degree. This is the reduced-homology register mechanism
behind the \KK{} qubit graph, called the qubit-register graph below. Second,
the positive spectrum on each join summand is obtained from sums of factor
eigenvalues, which lets complete-graph blocks act as local energetic penalties.
Here \(K_f\) denotes the complete graph on \(f\) vertices.

\begin{lemma}[Spectrum of the complete-graph clique Laplacian
\cite{gundert2015eigenvalues}]
\label{lem:complete-graph-laplacian}
The complete-complex spectrum gives the following full Hodge-Laplacian form for
the clique complex of a complete graph.
For \(1\le i\le f-1\),
\[
  \Delta_i(\Cl(K_f))=fI.
\]
In degree \(0\), the spectrum is \(0\) on the constant vector and \(f\) on
its orthogonal complement. Hence every nonconstant complete-graph component
has Hodge energy at least \(f\).
Under the augmented convention, the augmentation lifts the constant vector,
and
\[
  \DeltaAug_i(\Cl(K_f))=fI
  \qquad (-1\le i\le f-1).
\]
\end{lemma}

For reference, the up-Laplacian on \(C_i(\Cl(K_f))\) has eigenvalue \(0\) with
multiplicity \(\binom{f-1}{i}\) and eigenvalue \(f\) with multiplicity
\(\binom{f-1}{i+1}\). The down-Laplacian has the complementary spectrum:
eigenvalue \(f\) with multiplicity \(\binom{f-1}{i}\) and eigenvalue \(0\)
with multiplicity \(\binom{f-1}{i+1}\). See
\cite[Lem.~8]{gundert2015eigenvalues} for the complete-complex spectrum.

\subsection{\KK{} register and gadgets}

We recall only the structural part of the \KK{} construction
\cite{king2024gapped}; the gap estimates are imported later in
\Cref{sec:energy-analysis}. At the topological level, one logical qubit is
encoded by a bowtie graph \(\GQ=B_0\cup B_1\), where \(B_0\) and \(B_1\) are
the two squares representing the computational basis states \(\ket{0}\)
and \(\ket{1}\). The two cycles share one vertex. With the shared vertex
labelled \(v_4\), \(B_0\) is the join of two disconnected \(0\)-complexes on
\(\{v_1,v_3\}\) and \(\{v_2,v_4\}\), while \(B_1\) is the join of
\(\{v_4,v_6\}\) and \(\{v_5,v_7\}\).

For \(n\) qubits the \KK{} qubit graph, which we call the \emph{qubit-register graph},
is the join
\[
  \mathcal G_n=\GQ^{*n}.
\]
Its relevant Hodge degree is \(2n-1\). Choosing oriented one-cycle
representatives \(\regcycle{0}\) and \(\regcycle{1}\), supported on the squares
\(B_0\) and \(B_1\), respectively, the computational basis chain is
\[
  \Enc\ket{x_1\cdots x_n}
  =
  \regcycle{x_1}*\cdots *\regcycle{x_n}
  \in C_{2n-1}(\Cl(\mathcal G_n)).
\]
Thus the harmonic register is identified with the \(n\)-qubit Hilbert space.

Each local projector \(\phi_i\) specifies an integer linear combination of
these encoded basis cycles. The \KK{} filling gadget \(\mathcal T_i\) is
attached to the register clique complex so that this chosen cycle becomes
null-homologous, while the complementary encoded cycles are preserved by the
selectivity argument.

Flagness is part of the \KK{} construction. Thickening is realized by an
explicit graph~\cite[Lem.~7.3]{king2024gapped}, and the later coning and vertex
identifications preserve flagness~\cite[Lems.~8.3--8.4]{king2024gapped}.
Moreover, distinct gadget interiors have no edges between them, so no clique
spans two gadgets.

Let \(\GH\) be the resulting base graph. We therefore define the weighted
\KK{} base complex by
\[
  \GammaH
  =\Cl(\GH)
  =\Cl(\mathcal G_n)\cup_i \mathcal T_i
\,.
\]
The subscript \(H\) records the
input Hamiltonian instance whose local projector terms determine the attached
filling gadgets. The present paper later replaces this weighted base complex by
an unweighted blow-up, but the underlying register-and-filling picture is the
one just described.

\subsection{Complexity conventions}
\label{sec:complexity-conventions}

We work with \(\QMAone^{\Gtwo}\), where
\(\Gtwo=\{\mathsf{X},\mathsf{CX},\mathsf{CCX},H\otimes H\}\), following
Rudolph's exact-gate formulation; here \(\mathsf{CCX}\) is the Toffoli gate and
\(H\otimes H\) is treated as a two-qubit gate. The
hardness reduction starts from Rudolph's \(\Gtwo\) 4-QSAT construction
\cite[Thm.~4.2]{rudolph2025universala}, in the concrete form embedded into
clique homology in \cite[Thm.~6.8 and App.~D.1]{rudolph2025universala}. The
containment proof uses exact LCU simulation for sparse cyclotomic-integer
Hamiltonians.

The resulting source Hamiltonian is a frustration-free sum of rank-one
projectors onto integer states from a fixed finite local family. Here an
integer state means a normalized local state proportional to
\[
  \sum_{x\in\{0,1\}^m} z_x\ket{x},
  \qquad z_x\in\Z,
\]
for the computational basis on the \(m\)-qubit support of the projector. We
write \(\mloc\) for the maximum number of logical qubits touched by any source
projector; Rudolph's construction has \(\mloc\le4\). After weakening the stated
inverse-polynomial NO threshold if necessary, we write the source promise as
\[
  \pH=\frac{1}{q(n)},
\]
where \(q(n)\) is a positive integer-valued polynomial. This is the form
compatible with the \KK{} clique-homology gadgets: the gadget topology encodes
integer linear combinations of computational-basis cycles, whereas arbitrary
complex amplitudes are not directly available in a clique complex.

\section{Main Results}
\label{sec:main-results}

We first formulate the unweighted promise problem for clique homology.

\begin{definition}[Unweighted Gapped Clique Homology]
\label{def:unweighted-gapped-clique-homology}
Fix a positive integer \(c_{\mathrm{gap}}\) and put
\(\pgap(N)=N^{-c_{\mathrm{gap}}}\). An instance consists of an unweighted graph
\(G\) on \(N\) vertices and an integer \(0\le k<N\). Decide whether
\[
  \textsc{Yes}: H_k(\Cl(G);\R)\ne 0
\]
or
\[
  \textsc{No}: H_k(\Cl(G);\R)=0
  \quad\text{and}\quad
  \lambda_{\min}(\Delta_k(\Cl(G)))\ge \pgap(N)
\]
under the promise that either holds.
\end{definition}

Here \(\pgap(N)\) is the fixed promise gap of the target clique-homology
problem. We reserve \(\pH\) for the inverse-polynomial gap of the source
Hamiltonian instance.

\begin{theorem}[Main theorem]
\label{thm:main}
For a sufficiently large fixed constant \(c_{\mathrm{gap}}\), Unweighted
Gapped Clique Homology is \(\QMAone^{\Gtwo}\)-complete.
\end{theorem}

The hardness direction is formalized by the following gap-preserving reduction.
Containment is proved separately in \Cref{sec:containment}, using exact sparse
Hamiltonian simulation for the resulting integer clique Laplacian.

\begin{theorem}[Gap-preserving reduction]
\label{thm:gap-preserving-reduction}
Let \(H=\sum_{i=1}^t\phi_i\) be a Hamiltonian from the fixed finite family of
\(\Gtwo\)-integer local projectors recalled in
\Cref{sec:complexity-conventions}, with
\(\lambda_{\min}(H)=0\) in the YES case and
\(\lambda_{\min}(H)\ge \pH=1/q(n)\) in the NO case, where \(q(n)\) is a
positive integer-valued polynomial. The reduction outputs
an unweighted graph \(\bGH\), its clique complex
\(\bGammaH=\Cl(\bGH)\), and the integer \(k=2n-1\) such that
\[
  \lambda_{\min}(H)=0
  \Rightarrow H_{2n-1}(\bGammaH;\R)\ne0,
\]
whereas
\[
  \lambda_{\min}(H)\ge \pH
  \Rightarrow
  \lambda_{\min}(\bDelta_{2n-1}(\bGammaH))\ge \pgap(N)
\]
for the fixed promise \(\pgap(N)=N^{-c_{\mathrm{gap}}}\) in
\Cref{def:unweighted-gapped-clique-homology}.
\end{theorem}

\begin{proof}[Proof sketch]
The YES direction uses the weight-free selectivity of the \KK{} gadget and
homology preservation under the blow-up. The NO direction is the technical
gap-transfer theorem \Cref{thm:unweighted-no-gap}. The unweighted blow-up chain
space decomposes into the symmetric sector described in
\Cref{sec:blowup-reduction} and its orthogonal complement. On the symmetric
sector, the Laplacian is unitarily equivalent to \(M\) times the corresponding
weighted Laplacian, so the existing weighted many-gadget energy bound applies.
On the orthogonal complement, \Cref{thm:asymmetric-decoupling} gives a direct
lower bound. Thus the blow-up does not create additional low-energy degrees of
freedom. The intended weighted problem is recovered on the symmetric sector,
and all copy-label-dependent directions are energetically removed. Combining
these two bounds gives the promised inverse-polynomial gap.
\end{proof}

\begin{proof}[Proof of \Cref{thm:main}]
\Cref{thm:gap-preserving-reduction} gives \(\QMAone^{\Gtwo}\)-hardness. The
explicit choices of \(\ratio\) and \(M\) in the proof of
\Cref{thm:unweighted-no-gap} make the output size polynomial and give a fixed
bound \(N^{-c_{\mathrm{gap}}}\) for a sufficiently large constant
\(c_{\mathrm{gap}}\). The containment statement is
\Cref{thm:containment}. Together they give completeness.
\end{proof}

\section{Blow-up and Symmetric Reduction}
\label{sec:blowup-reduction}

This section is independent of the \KK{} construction. We work with an
arbitrary clique complex \(\Gamma\) equipped with a level function on its
vertices, and record the formal properties of the associated unweighted
blow-up.

The level function records how the intended weighted metric should scale from
one vertex block to the next.
If \(v\) has level \(\ell(v)\), then the blow-up gives it block
multiplicity \(\bfmult_v=M/\ratio^{\ell(v)}\). The
arguments below depend only on the clique-complex structure, the complete
multipartite blow-up, and the copy-relabeling symmetry inside each block.

\subsection{The level-ratio blow-up}

\begin{definition}[Level-ratio blow-up]
\label{def:level-ratio-blowup}
Let \(\Gamma=\Cl(G)\) be a base clique complex of a graph \(G=(V,E)\) with a level function
\(\level:V\to \Z_{\ge0}\), and write
\[
  L=\max_{v\in V}\level(v).
\]
Choose integers \(\ratio\ge2\) and \(M\ge1\) such that \(M\) is an integer
multiple of \(\ratio^L\). Equivalently, \(M=\widetilde M\ratio^L\) for some
\(\widetilde M\in\Z_{\ge1}\). Set
\[
  \bfmult_v=\bfmult(\level(v))=\frac{M}{\ratio^{\level(v)}}.
\]
The blow-up graph replaces each base vertex \(v\) by a clique \(K_{\bfmult_v}\), and
each base edge \(uv\) by the complete bipartite graph between the corresponding
blocks \(K_{\bfmult_u}\) and \(K_{\bfmult_v}\). The output complex is the clique complex of this graph.
When no special name is needed, we denote the output complex by \(\bGamma\).
Every \(\bfmult_v\) is then a positive integer, and the smallest block has size
\(\min_v \bfmult_v=M/\ratio^L=\widetilde M\).
\end{definition}

We write \(\bd_k:C_k(\bGamma)\to C_{k+1}(\bGamma)\) for the unweighted
coboundary of the blow-up, and
\(\bpartial_{k+1}=(\bd_k)^\dagger\) for its adjoint in the ordinary
simplex-basis inner product. The blow-up Hodge Laplacian is
\[
  \bDelta_k
  =
  \bpartial_{k+1}\bd_k+\bd_{k-1}\bpartial_k.
\]
The multiplicity ratio is chosen so that
\[
  \frac{\bfmult(l+1)}{\bfmult(l)}=\frac{1}{\ratio}.
\]
Recall from \Cref{sec:preliminaries} that the weighted orthonormal simplex
basis is obtained by rescaling
\[
  \ket{\sigma}'=w(\sigma)^{-1}\ket{\sigma},
  \qquad
  w(\sigma)=\prod_{v\in\sigma}w(v).
\]
With this convention, adjoining a vertex \(v\) contributes the factor \(w(v)\)
to a coboundary matrix element. In the blow-up, the symmetric \(0\)-chain on a
level-\(l\) block is
\[
  \ket{\widehat{\bm v}}
  =
  \frac{1}{\sqrt{\bfmult(l)}}\sum_{i=1}^{\bfmult(l)}\ket{v_i}.
\]
Under the symmetric-sector identification, the induced amplitude weight is
\(\sqrt{\bfmult(l)}\). The top level has induced weight \(\sqrt M\), whereas
the corresponding weighted vertex has weight \(1\). Thus, when comparing with
the weighted profile, we remove this common top-level factor. The resulting
relative symmetric weight is
\[
  \sqrt{\frac{\bfmult(l)}{M}}
  =
  \ratio^{-l/2}.
\]
Thus one level step changes the induced vertex weight by
\(\ratio^{-1/2}\). In the weighted gadget notation this factor is
\(\lambda\), so \(\lambda^2=1/\ratio\). In the later application, \(L\) depends only on the fixed local
gadget family and is therefore constant. Choosing \(M\) polynomially large then
keeps the whole blow-up polynomial size.

\begin{table}[H]
  \centering
  \begin{tabular}{c|c|c|c}
    vertex level \(l\) & block size \(\bfmult(l)\) & relative symmetric weight
    \(\sqrt{\bfmult(l)/M}\) & weighted parameter \\
    \hline
    \(0\) & \(M\) & \(1\) & \(1\) \\
    \(1\) & \(M/\ratio\) & \(\ratio^{-1/2}\) & \(\lambda\) \\
    \(2\) & \(M/\ratio^2\) & \(\ratio^{-1}\) & \(\lambda^2\) \\
    \(\vdots\) & \(\vdots\) & \(\vdots\) & \(\vdots\) \\
    \(l\) & \(M/\ratio^l\) & \(\ratio^{-l/2}\) & \(\lambda^l\)
  \end{tabular}
  \caption{Level correspondence for the blow-up. The third column records the
  symmetric-sector amplitude weight after dividing by the common top-level
  factor \(\sqrt M\). Increasing the vertex level by one divides the block size
  by \(\ratio\), which gives exactly the weight step
  \(\lambda=\ratio^{-1/2}\). The \KK{} application below uses only levels \(0\)
  and \(1\).}
  \label{tab:level-correspondence}
\end{table}

\Cref{tab:level-correspondence} summarizes the relation between vertex levels,
block multiplicities, and the induced weighted parameters.

The construction should be read as replacing a weighted vertex by a block of
indistinguishable copies. 
The raw blow-up is unweighted: every simplex basis vector has norm \(1\). The
weighted metric appears only after restricting to the fully symmetric averages
over these copy choices. \Cref{fig:blowup-correspondence} illustrates this
replacement on a small example.

\begin{figure}[t]
  \centering
  \includegraphics[width=\textwidth]{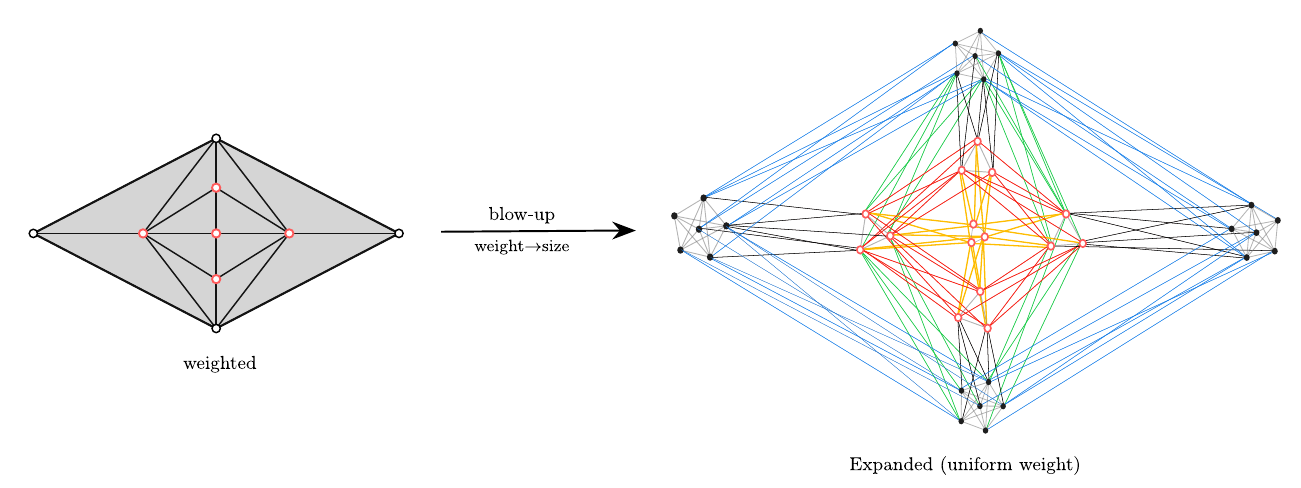}
  \caption{The blow-up correspondence. A weighted base complex is replaced by an
  unweighted expanded complex in which vertex weights are represented by block
  sizes. For simplicity, only the expanded vertices and incomplete edges are drawn.
  Symmetric averages over one-copy-per-block choices reconstruct the weighted
  Hodge metric. The illustration uses block sizes \(6\) and \(3\), exhibiting
  an integral ratio of \(2\).}
  \label{fig:blowup-correspondence}
\end{figure}

\subsection{Symmetric sector}

Each vertex \(v\) of the base complex has been replaced by a clique block
\(\widehat{\bm B}_v\) of \(\bfmult_v\) indistinguishable copies. Thus the
blow-up carries a natural copy-relabeling symmetry: we may rename the vertices
inside \(\widehat{\bm B}_v\) without changing the underlying construction. The
sector relevant to the weighted base complex is the part that does not see
these labels, namely the span of uniform averages over all copy choices.

Here ``symmetric'' means only ``unchanged when copies inside a block are
renamed.'' The argument below uses ordinary orthogonal projections that average
over these renamings.

Formally, let \(\BlockSym=\prod_{v\in V} S_{\widehat{\bm B}_v}\) act by
renaming copies inside each block. We write \(\bT_\bullet\) for the subspace fixed by this action
and \(\bT_\bullet^\perp\) for its orthogonal complement. The bullet denotes
the graded collection over all chain degrees. Thus \(\bT_\bullet\) is simply
the subspace fixed by every copy relabeling.

For a base vertex \(v\), write
\[
  \ket{\widehat{\bm v}}
  =
  \frac{1}{\sqrt{\bfmult_v}}\sum_{v_i\in\widehat{\bm B}_v}\ket{v_i}
\]
for the uniform \(0\)-chain on \(\widehat{\bm B}_v\). For an oriented base
simplex \(\sigma=(v_0,\ldots,v_k)\) and a copy choice
\(\mathbf i=(i_0,\ldots,i_k)\), write
\[
  \widetilde\sigma_{\mathbf i}
  =
  \bigl((v_0)_{i_0},\ldots,(v_k)_{i_k}\bigr)
\]
for the corresponding oriented simplex in the blow-up. Define the
one-copy-per-block average of \(\sigma\), meaning the average over simplices
that use exactly one copy from each block appearing in \(\sigma\), by
\[
  \ket{\widehat{\bm{\sigma}}}
  =
  \ket{\widehat{\bm v_0}}*\cdots *\ket{\widehat{\bm v_k}}
  =
  \frac{1}{\sqrt{\prod_{v\in\sigma}\bfmult_v}}
  \sum_{i_0,\ldots,i_k}
  \ket{\widetilde\sigma_{\mathbf i}},
\]
where the sum ranges over choices of one copy from each block. For example, an
edge \(uv\) of the base graph becomes the complete bipartite family of edges
between the \(u\)-block and the \(v\)-block, and its symmetric representative is
\[
  \ket{\widehat{\bm{uv}}}
  =
  \frac{1}{\sqrt{\bfmult_u\bfmult_v}}
  \sum_{i=1}^{\bfmult_u}\sum_{j=1}^{\bfmult_v}\ket{u_i v_j}.
\]
These are the chains that remember the base simplex but forget the copy labels.

Only these one-copy-per-block averages survive in the fixed subspace. If a
simplex contains two copies from the same block, swapping those two copies
fixes the underlying vertex set but reverses the orientation. Hence its average
over all copy renamings is equal to its negative, and is therefore zero. Thus
\[
  \bT_k
  =
  \operatorname{span}\{\ket{\widehat{\bm{\sigma}}}:\sigma\in\Gamma_k\}.
\]
In particular, repeated use of one block creates no additional symmetric
sector.

\begin{lemma}[Symmetric reduction]
\label{lem:symmetric-reduction}
There is an isometric isomorphism
\[
  \Phi:\bT_\bullet \longrightarrow C_\bullet(\Gamma,\wtilde),
  \qquad
  \wtilde(v)=\sqrt{\bfmult_v},
\]
where \(C_\bullet(\Gamma,\wtilde)\) denotes the weighted simplex Hilbert space from
\Cref{sec:preliminaries}. Let \(\partial_{\wtilde}\) and
\(\Delta_{\wtilde}\) denote the weighted adjoint boundary and Hodge Laplacian
on \(\Gamma\), degree by degree. Then \(\Phi\) intertwines the algebraic
coboundary and its weighted adjoint, and hence
\[
  \Phi\,\bDelta_k\big|_{\bT_k}\,\Phi^{-1}
  =
  \Delta_{\wtilde,k}
  \qquad\text{for every }k.
\]
\end{lemma}

\begin{proof}
The vectors \(\ket{\widehat{\bm{\sigma}}}\) form an orthonormal basis of
\(\bT_k\). Send \(\ket{\widehat{\bm{\sigma}}}\) to the normalized base vector
\[
  \ket{\sigma}'_{\wtilde}
  =
  \wtilde(\sigma)^{-1}\ket{\sigma}
  =
  \left(\prod_{v\in\sigma}\bfmult_v\right)^{-1/2}\ket{\sigma}
\]
in the weighted degree-\(k\) simplex space \(C_k(\Gamma,\wtilde)\). This defines an isometry
\(\Phi\).

Fix one incidence \(\tau=\sigma\cup\{b\}\), and put
\(F_\sigma=\prod_{v\in\sigma}\bfmult_v\). With the notation above, write
\(\widetilde\tau_{\mathbf i,j}=\widetilde\sigma_{\mathbf i}\cup\{b_j\}\).
We now check the coboundary in two parts. First, the component along the
ordinary base coface \(\tau\) has the weighted coefficient expected from the
base complex. Second, all cofaces that use the same base block twice cancel,
so no extra symmetric component remains.
Let \(\Pi_\tau\) be the projection onto the span of the one-copy-per-block
copies of \(\tau\). Then the \(\tau\)-component of the blow-up coboundary is
\[
  \Pi_\tau\bd_k\ket{\widehat{\bm{\sigma}}}
  =
  [\tau:\sigma]\frac{1}{\sqrt{F_\sigma}}
  \sum_{\mathbf i}\sum_{j=1}^{\bfmult_b}\ket{\widetilde\tau_{\mathbf i,j}}
  =
  [\tau:\sigma]\sqrt{\bfmult_b}\ket{\widehat{\bm{\tau}}}.
\]
Taking the matrix coefficient between normalized orbit averages gives
\[
  \bra{\widehat{\bm{\tau}}}\,\bd_k\,
  \ket{\widehat{\bm{\sigma}}}
  =
  [\tau:\sigma]\sqrt{\bfmult_b}.
\]
On the weighted base complex, the coboundary-first convention from
\Cref{sec:preliminaries} gives exactly the same coefficient:
\[
  \bra{\tau}'_{\wtilde}d_k\ket{\sigma}'_{\wtilde}
  =
  [\tau:\sigma]\frac{\wtilde(\tau)}{\wtilde(\sigma)}
  =
  [\tau:\sigma]\sqrt{\bfmult_b}.
\]

It remains to check that no coface using the same base block twice contributes
to the symmetric sector. Suppose such a coface \(\eta\) contains two copies
\(a_r,a_s\) from the same block. Then \(\eta\) can arise in two ways, by adding
\(a_r\) to \(\eta\setminus\{a_r\}\) or by adding \(a_s\) to
\(\eta\setminus\{a_s\}\). These two contributions have opposite incidence signs:
\[
  [\eta:\eta\setminus\{a_r\}]
  =
  -[\eta:\eta\setminus\{a_s\}],
\]
because the transposition of \(a_r\) and \(a_s\) fixes the underlying
unoriented simplex and reverses the orientation. Hence their sum is zero.
Thus repeated-block cofaces do not appear in the symmetric sector, and
\(\Phi\bd_k\Phi^{-1}=d_k\).

Since \(\Phi\) is an isometry, taking adjoints in
\(\Phi\bd_k\Phi^{-1}=d_k\) gives
\[
  \Phi\bpartial_{k+1}\Phi^{-1}
  =
  \partial_{\wtilde,k+1}.
\]
Thus the up and down terms agree separately, and for all base \(k\)-simplices
\(\sigma,\sigma'\),
\[
  \bra{\widehat{\bm{\sigma'}}}\bDelta_k
  \ket{\widehat{\bm{\sigma}}}
  =
  \bra{\sigma'}'_{\wtilde}\Delta_{\wtilde,k}
  \ket{\sigma}'_{\wtilde}.
\]
This is precisely
\(\Phi\,\bDelta_k\big|_{\bT_k}\,\Phi^{-1}=\Delta_{\wtilde,k}\).
\end{proof}

We now isolate the global rescaling introduced by the top block size \(M\).
For the block multiplicities \(\bfmult_v=M\ratio^{-\level(v)}\), put
\(w_\ratio(v)=\ratio^{-\level(v)/2}\). Then
\[
  \wtilde(v)=\sqrt{\bfmult_v}
  =\sqrt M\,\ratio^{-\level(v)/2}
  =\sqrt M\,w_\ratio(v).
\]
Every coboundary incidence adjoins exactly one vertex, so the matrix of the
same algebraic coboundary \(d_k\) in the respective orthonormal simplex bases is
rescaled by \(\sqrt M\):
\[
  \bra{\tau}'_{\wtilde}d_k\ket{\sigma}'_{\wtilde}
  =
  \sqrt M\,
  \bra{\tau}'_{w_\ratio}d_k\ket{\sigma}'_{w_\ratio}.
\]
Taking adjoints gives the same rescaling for the weighted boundaries
\(\partial_{\wtilde}\) and \(\partial_{w_\ratio}\).
Both the up- and down-Laplacians are therefore multiplied by \(M\), so
\[
  \Delta_{\wtilde,k}=M\Delta_{w_\ratio,k}.
\]
In the later \KK{} application, \(w_\ratio(v)=\lambda^{\level(v)}\) with
\(\lambda=1/\sqrt{\ratio}\).

Thus the top size \(M\) is not discarded as a harmless normalization. It gives
an exact common factor \(M\) in the symmetric-sector Laplacian, while adjacent
levels encode the perturbative parameter
\(\lambda=1/\sqrt{\ratio}\).

\subsection{Gap outside the symmetric sector}

Let \(\Gamma=\Cl(G)\) be the finite base clique complex with vertex set \(V\), and
let \(\bGamma\) be its unweighted blow-up. We write
\(\widehat{\bm V}=V(\bGamma)\) for the expanded vertex set. For each
\(v\in V\), let \(\widehat{\bm B}_v\subseteq\widehat{\bm V}\) be the clique
block defined above, so \(|\widehat{\bm B}_v|=\bfmult_v\). The copy relabeling group
\[
  \BlockSym=\prod_{v\in V}S_{\widehat{\bm B}_v}
\]
acts on \(C_k(\bGamma;\R)\) by permuting copies inside each block.  Let
\(\bT_k=C_k(\bGamma;\R)^{\BlockSym}\) be the symmetric sector defined above,
let \(\bT_k^\perp\) denote its orthogonal complement in the ordinary
simplex-basis inner product, and let \(\bDelta_k\) be the unweighted Hodge
Laplacian of \(\bGamma\).

\begin{theorem}[Gap outside the symmetric sector]
\label{thm:asymmetric-decoupling}
For every chain degree \(k\),
\[
  \bDelta_k\big|_{\bT_k^\perp}\succeq \min_{v\in V} \bfmult_v.
\]
\end{theorem}

\begin{proof}
Let \(\widehat C_k=C_k(\bGamma;\R)\). For each base vertex \(v\), define the
block-averaging operator
\[
  \widehat{\bm A}_v
  =
  \frac{1}{\bfmult_v!}\sum_{\pi\in S_{\widehat{\bm B}_v}}U_\pi,
\]
where \(U_\pi\) renames the copies in \(\widehat{\bm B}_v\). This is the orthogonal
projector onto chains unchanged by those renamings. The operators
\(\widehat{\bm A}_v\)
commute with one another and with the coboundary, its adjoint, and
\(\bDelta_k\). Their product is the projector onto \(\bT_k\).

Order the base vertices as \(v_1,\ldots,v_R\) where \(|V|=R\), and set
\[
  \widehat{\bm Q}_j
  =
  \widehat{\bm A}_{v_1}\cdots
  \widehat{\bm A}_{v_{j-1}}(I-\widehat{\bm A}_{v_j}).
\]
The \(\widehat{\bm Q}_j\) are pairwise orthogonal projectors, they commute with
\(\bDelta_k\). Indeed, for \(j<\ell\), the product
\(\widehat{\bm Q}_j\widehat{\bm Q}_\ell\) contains
\((I-\widehat{\bm A}_{v_j})\widehat{\bm A}_{v_j}\) and is therefore zero.
The elementary telescoping identity
\[
  I-\prod_{j=1}^R \widehat{\bm A}_{v_j}
  =
  \sum_{j=1}^R \widehat{\bm Q}_j
\]
therefore decomposes \(\bT_k^\perp\) into their images as
\[
  \bT_k^\perp
  =
  \bigoplus_{j=1}^R \im\widehat{\bm Q}_j.
\]
We call \(\im\widehat{\bm Q}_j\) the \(j\)-block asymmetric sector. It is
therefore enough to bound one such sector, with \(j\) fixed.
Note that every
\(\psi\in\im \widehat{\bm Q}_j\) satisfies
\(\widehat{\bm A}_{v_j}\psi=0\): its average over copy names in
\(\widehat{\bm B}_{v_j}\) vanishes.

Consider the degree-\(k\) span of simplices disjoint from the expanded block
\(\widehat{\bm B}_{v_j}\),
\[
  \operatorname{span}
  \bigl\{
    |\sigma\rangle :
    \sigma\in\bGamma,\ \dim\sigma=k,\ \sigma\cap\widehat{\bm B}_{v_j}=\varnothing
  \bigr\}.
\]
Relabeling copies inside \(\widehat{\bm B}_{v_j}\) does nothing on this subspace,
so \(\widehat{\bm A}_{v_j}\) is the identity there. Since
\(\widehat{\bm A}_{v_j}\) is self-adjoint, \(\widehat{\bm A}_{v_j}\psi=0\) makes
\(\psi\) orthogonal to the above subspace. Equivalently, \(\psi\) is
supported on simplices that meet \(\widehat{\bm B}_{v_j}\).

For the fixed base vertex \(v_j\), let
\(L_{v_j}:=\mathrm{lk}_{\Gamma}(v_j)\) be its link in the base complex, explicitly
\[
  L_{v_j}
  =
  \{\tau\in\Gamma:\tau\cap\{v_j\}=\varnothing,\ \tau\cup\{v_j\}\in\Gamma\}.
\]
Let
\(\widehat{\bm L}_{v_j}\) be the blow-up of \(L_{v_j}\) inside \(\bGamma\);
it is supported on expanded vertices outside the block
\[
  V(\widehat{\bm L}_{v_j})
  \subseteq
  \widehat{\bm V}\setminus\widehat{\bm B}_{v_j}.
\]
Also let \(K_{\widehat{\bm B}_{v_j}}\) be the complete graph on the expanded
block \(\widehat{\bm B}_{v_j}\).
Let
\[
  \widehat{\bm S}_{v_j}
  =
  \Cl(K_{\widehat{\bm B}_{v_j}})*\widehat{\bm L}_{v_j}
  \subseteq \bGamma
\]
be the expanded star around the block \(\widehat{\bm B}_{v_j}\), namely the
subcomplex generated by simplices that contain at least one copy from that
block, together with their faces. Thus the first factor is the simplex on the
copies of the base vertex \(v_j\), while the second factor is the blow-up of the
base link \(L_{v_j}\).

The point of introducing \(\widehat{\bm S}_{v_j}\) is locality of the energy
calculation. The preceding support statement says that every
\(\psi\in\im\widehat{\bm Q}_j\) is supported on \(k\)-simplices that meet
\(\widehat{\bm B}_{v_j}\):
\[
  \sigma\in\operatorname{supp}(\psi)
  \quad\Longrightarrow\quad
  \sigma\cap\widehat{\bm B}_{v_j}\ne\varnothing .
\]
Hence \(\sigma\in\widehat{\bm S}_{v_j}\). Moreover,
\[
  \tau\subseteq\sigma
  \quad\Longrightarrow\quad
  \tau\in\widehat{\bm S}_{v_j},
  \qquad
  \sigma\subseteq\eta\in\bGamma
  \quad\Longrightarrow\quad
  \eta\cap\widehat{\bm B}_{v_j}\ne\varnothing
  \quad\Longrightarrow\quad
  \eta\in\widehat{\bm S}_{v_j}.
\]
Thus applying the boundary or coboundary to \(\psi\) sees no simplex outside
\(\widehat{\bm S}_{v_j}\):
\[
  \partial_k\psi\in C_{k-1}(\widehat{\bm S}_{v_j}),
  \qquad
  d_k\psi\in C_{k+1}(\widehat{\bm S}_{v_j}).
\]
Therefore the quadratic form of \(\bDelta_k\) on \(\psi\) agrees with the
quadratic form of the Hodge Laplacian on \(\widehat{\bm S}_{v_j}\).

We have reduced the estimate to the expanded star of the block where the
asymmetry is detected. The block \(\widehat{\bm B}_{v_j}\) is the source of the
energy lower bound, while \(\widehat{\bm S}_{v_j}\) is the face-closed
subcomplex needed to compute the boundary and coboundary contributions.
The join decomposition of \(\widehat{\bm S}_{v_j}\) now separates these two
roles.

By \Cref{lem:horak-jost}, the augmented Laplacian on
\(\widehat{\bm S}_{v_j}\) is the tensor sum of the augmented Laplacians on
\(\Cl(K_{\widehat{\bm B}_{v_j}})\) and on the expanded-link factor. More
explicitly,
\[
  C_k(\widehat{\bm S}_{v_j})
  =
  \bigoplus_{r+s=k-1}
  C_r(\Cl(K_{\widehat{\bm B}_{v_j}}))\otimes
  C_s(\widehat{\bm L}_{v_j}).
\]
The copy-renaming operator \(\widehat{\bm A}_{v_j}\) acts only on the first tensor
factor. On the degree-\(-1\) factor it is the identity, so
\(\widehat{\bm A}_{v_j}\psi=0\) removes that summand. On every remaining summand
\(r\ge0\),
\Cref{lem:complete-graph-laplacian} says that the augmented Laplacian of the
complete simplex \(\Cl(K_{\widehat{\bm B}_{v_j}})\) is \(\bfmult_{v_j}I\), while the
expanded-link Laplacian is positive semidefinite.
Write \(V_j^{\mathrm{loc}}=\{v_j\}\cup V(L_{v_j})\) for the base vertices whose
blocks occur in \(\widehat{\bm S}_{v_j}\).
Hence the augmented Laplacian on \(\widehat{\bm S}_{v_j}\) satisfies, for every
\(\psi\in\im\widehat{\bm Q}_j\),
\[
  \langle\psi,\Delta_k^{\mathrm{aug},\widehat{\bm S}_{v_j}}\psi\rangle
  \ge \bfmult_{v_j}\|\psi\|^2
  \ge \min_{u\in V_j^{\mathrm{loc}}}\bfmult_u\,\|\psi\|^2
  \ge \min_{v\in V}\bfmult_v\,\|\psi\|^2.
\]
For \(k>0\), this is the ordinary Laplacian. When \(k=0\), the augmented
operator differs by the rank-one term \(d_{-1}(d_{-1})^\dagger\), whose image is
the constant \(0\)-chain. This constant chain lies in the symmetric sector for
every block, whereas every vector in \(\im\widehat{\bm Q}_j\) lies in
\(\ker\widehat{\bm A}_{v_j}\). Hence the augmentation term vanishes on
\(\im\widehat{\bm Q}_j\). Thus the same lower bound holds for the ordinary
Laplacian, and applying it on every \(\im\widehat{\bm Q}_j\) proves the claim.
\end{proof}

The proof is local in the chosen block \(v_j\). In particular, the link
\(\widehat{\bm L}_{v_j}\) may be combinatorially complicated and only
positivity of its Laplacian is used. This locality is what makes the same
decoupling applicable after many local attachments are added to the base
complex.

\begin{corollary}[No spurious homology]
\label{cor:no-spurious-homology}
The blow-up creates no homology outside the symmetric sector. In particular,
\[
  H_k(\bGamma;\R)\cong H_k(\Gamma;\R)
\]
after identifying the symmetric sector with the weighted base complex.
\end{corollary}

\begin{proof}
By Hodge theory, homology is the kernel of the Hodge Laplacian.
\Cref{thm:asymmetric-decoupling} shows that \(\bT_\bullet^\perp\) has no zero modes. All harmonic chains therefore
lie in \(\bT_\bullet\), where \Cref{lem:symmetric-reduction} identifies them
with harmonic chains of the weighted base complex. Changing positive weights
does not change the underlying homology groups.
\end{proof}

\section{Expanded Gadget Complex}
\label{sec:expanded-gadget-complex}

This section applies the general blow-up formalism of
\Cref{sec:blowup-reduction} to the gadget complex used in the hardness
reduction. We first recall the weighted base complex, then build its unweighted
expanded qubit register and attach the expanded local gadgets term by term.
The goal is to produce the unweighted clique complex
\(\bGammaH=\Cl(\bGH)\) whose symmetric sector reproduces the weighted gadget
complex, while the nonsymmetric directions are separated by the decoupling
bound proved above. The spectral estimates for this final complex are carried
out in \Cref{sec:energy-analysis}.

\subsection{Base complex}

Let \(\GH\) be the base gadget graph of \cite{king2024gapped} recalled in
\Cref{sec:preliminaries}, and let
\[
  \GammaH=\Cl(\GH)=\Cl(\mathcal G_n)\cup_i\mathcal T_i.
\]
For distinct local terms \(i\ne j\), the interiors
\(\mathcal T_i\setminus \Cl(\mathcal G_n)\) and
\(\mathcal T_j\setminus \Cl(\mathcal G_n)\) are disjoint and have no edges
between them. Thus the union introduces no unintended mixed cliques. The
selectivity part of the construction is weight-free and therefore applies
before any spectral estimate is invoked.

\subsection{Qubit register}
\label{sec:qubit-register}

We first isolate the qubit-register part, since it is the stable harmonic
subspace on which the local projectors act. Register vertices have level zero,
so their expanded blocks have the top size \(M\). Let \(K_a\) denote the copy
of \(K_M\) replacing the \(a\)th bowtie vertex. In the one-qubit register there
are seven such blocks
\(K_1,\ldots,K_7\): the two squares of the bowtie share the central block
\(K_4\). Write
\[
  \hregvertex{a}=\frac{1}{\sqrt{M}}\sum_{v\in V(K_a)}\ket{v}
\]
for the normalized uniform \(0\)-chain on its vertices. The expanded
one-qubit bowtie has two register subcomplexes
\[
  \bR_0(M)=(K_1\sqcup K_3)*(K_2\sqcup K_4),
  \qquad
  \bR_1(M)=(K_4\sqcup K_6)*(K_5\sqcup K_7),
\]
and
\[
  \bGQ(M)=\bR_0(M)\cup \bR_1(M).
\]
The two expanded encoded one-qubit cycles are
\[
  \ket{\hregcycle{0}}
  =
  \frac12(\hregvertex{1}-\hregvertex{3})*
  (\hregvertex{2}-\hregvertex{4}),
  \qquad
  \ket{\hregcycle{1}}
  =
  \frac12(\hregvertex{4}-\hregvertex{6})*
  (\hregvertex{5}-\hregvertex{7}).
\]
For \(n\) qubits, define the expanded qubit-register graph by the join product
\[
  \bGcal_n(M)=\bGQ(M)^{*n}.
\]
Its target Hodge degree is \(2n-1\), matching the join of \(n\) one-dimensional
cycle representatives. The one-qubit expanded bowtie and the \(n\)-qubit join
register are shown schematically in \Cref{fig:expanded-qubit-gadget}.

\begin{figure}[t]
  \centering
  \includegraphics[width=.86\textwidth]{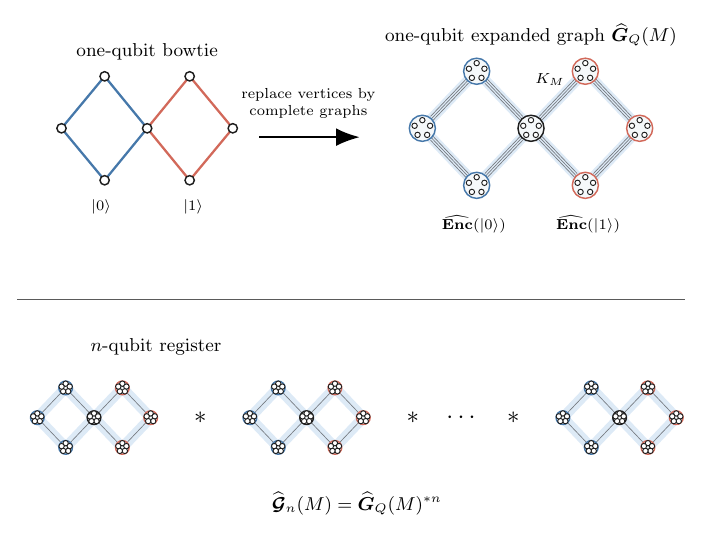}
  \caption{The qubit-register blow-up. A one-qubit bowtie has two harmonic
  cycles representing \(|0\rangle\) and \(|1\rangle\). In the expanded register,
  each vertex is replaced by a complete graph and each bowtie edge by the
  corresponding complete bipartite connection. The \(n\)-qubit register is the
  join product of these one-qubit expanded graphs. Each register block has size
  \(M=\rho\); the displayed copy counts are schematic.}
  \label{fig:expanded-qubit-gadget}
\end{figure}

\begin{definition}[Register encoding maps]
\label{def:register-encoding}
Fix once and for all an orientation convention for joins. For the unexpanded
\KK{} register, label the seven bowtie vertices \(v_1,\ldots,v_7\) and set
\[
  \ket{\regcycle{0}}
  =
  \frac12\bigl(
    \ket{v_1v_2}+\ket{v_2v_3}
    +\ket{v_3v_4}+\ket{v_4v_1}
  \bigr),
\]
\[
  \ket{\regcycle{1}}
  =
  \frac12\bigl(
    \ket{v_4v_5}+\ket{v_5v_6}
    +\ket{v_6v_7}+\ket{v_7v_4}
  \bigr),
\]
where the oriented edges are ordered according to the chosen cycle
orientation. The unexpanded register encoding is the linear map
\[
  \Enc:(\C^2)^{\otimes n}\longrightarrow C_{2n-1}(\Cl(\mathcal G_n)),
  \qquad
  \Enc\ket{x_1\cdots x_n}
  =
  \ket{\regcycle{x_1}}*\cdots *\ket{\regcycle{x_n}}.
\]
For the expanded register, replace each representative vertex by the uniform
block vector \(\hregvertex{a}\) and define
\[
  \hEnc\ket{x_1\cdots x_n}
  =
  \ket{\hregcycle{x_1}}*\cdots *\ket{\hregcycle{x_n}}
  \in C_{2n-1}(\Cl(\bGcal_n(M))).
\]
Equivalently, under the standard chain identification for joins, the displayed
join product is the tensor product
\(\ket{\hregcycle{x_1}}\otimes\cdots\otimes\ket{\hregcycle{x_n}}\), up to the
fixed global sign convention. The expanded encoding has the useful property that
\(\hEnc\ket{x}\) is the fully symmetric average over all copy
choices in the blocks, while \(\Enc\ket{x}\) is the corresponding
minimal representative before expansion.
\end{definition}

This distinction anticipates the symmetric-sector reduction. The blow-up does
not change which logical cycle is being represented; it replaces a selected
representative chain by its block-symmetric harmonic representative. Thus the
logical Hilbert space is carried into the symmetric sector of the expanded
chain space.

\begin{lemma}[Expanded qubit register]
\label{lem:expanded-qubit-register}
There are universal constants \(c,C>0\) such that, on
\(\Cl(\bGcal_n(M))\),
\[
  \ker \bDelta_{2n-1}
  =
  \im(\hEnc),
\]
and, on the orthogonal complement of this kernel in degree \(2n-1\),
\[
  \bDelta_{2n-1}\succeq c\,M I,
  \qquad
  \|\bDelta_{2n-1}\|\le C\,nM.
\]
\end{lemma}

\begin{proof}
Let \(\mathcal Q=\Cl(\GQ)\) be the unexpanded seven-vertex bowtie complex.
It is connected, its reduced first homology is two-dimensional, generated by
\(\ket{\regcycle{0}}\) and \(\ket{\regcycle{1}}\), and all its other reduced
homology groups vanish. Since \(\mathcal Q\) is fixed, its
\(\DeltaAug\)-operators have a smallest positive eigenvalue
\(\gamma_{\mathcal Q}>0\), independent of \(M\).

The complex \(\Cl(\bGQ(M))\) is exactly the blow-up of \(\mathcal Q\) with
\(\bfmult_v=M\) for all seven base vertices. On the symmetric sector,
\Cref{lem:symmetric-reduction} therefore gives
\[
  \widehat{\bm{\Delta}}^{\mathrm{aug}}\big|_{\bT}
  \cong M\DeltaAug_{\mathcal Q}.
\]
On its orthogonal complement, \Cref{thm:asymmetric-decoupling} gives
\(\widehat{\bm{\Delta}}\succeq M I\) in every nonnegative degree. This also
holds for the augmented operator: in positive degrees the two operators agree,
in degree \(0\) the augmented operator only adds a positive semidefinite
rank-one term, and in degree \(-1\) there is no asymmetric subspace. Hence the
only
augmented harmonic space of the expanded one-qubit register is in degree
\(1\), where it is generated by
\(\ket{\hregcycle{0}}\) and \(\ket{\hregcycle{1}}\), and every nonzero
one-qubit-register eigenvalue is at least
\(\min\{\gamma_{\mathcal Q},1\}M\).

The one-qubit operator norm has the same linear scale. We use only the crude
bound that, for an unweighted clique complex, the Hodge-Laplacian norm is at
most twice the number of vertices. Since the expanded one-qubit register has
\(|V(\bGQ(M))|=7M\) vertices, its Hodge operators have norm at most \(14M\).

Now apply the augmented join formula
\Cref{lem:horak-jost} to
\(\bGcal_n(M)=\bGQ(M)^{*n}\). The Laplacian on every join summand is a tensor
sum of the \(n\) one-qubit augmented Laplacians. Because the one-qubit
augmented kernel
occurs only in degree \(1\), the degree-\((2n-1)\) join kernel is precisely the
tensor product of \(n\) copies of that two-dimensional space, namely
\(\im(\hEnc)\). Any vector orthogonal to this kernel has a positive-energy
factor and hence energy at least
\(cM\), where \(c=\min\{\gamma_{\mathcal Q},1\}\). The tensor-sum norm is at
most \(14nM\), so the asserted upper bound holds with \(C=14\).
Since \(2n-1>0\), the augmented operator in the target degree is the ordinary
Hodge Laplacian stated in the lemma.
\end{proof}

\subsection{Filling gadgets}
\label{sec:filling-gadgets}

The harmonic space of the qubit-register graph is identified with the logical
\(n\)-qubit Hilbert space by choosing cycle representatives for computational
basis states. Each local term then specifies an integer cycle on this graph,
and the corresponding gadget is attached to fill exactly that cycle. In the
resulting complex, the sets of simplices
\(\mathcal T_i\setminus \Cl(\mathcal G_n)\) are pairwise disjoint, and the
gadgets meet one another only through the qubit-register graph. This is the
structural reason that the later many-gadget estimate can treat the local
projector gadgets separately before controlling their shared down-Laplacian
contribution.

\paragraph{Binary vertex weights.}
The \KK{} construction assigns weight \(1\) to all vertices of the original
qubit graph and weight \(\lambda\) to every added gadget vertex
\cite[Sec.~8]{king2024gapped}. In the explicit single-gadget construction, the
outer layer identified with the qubit graph retains weight \(1\), while the
inside layer and central vertex receive weight \(\lambda\)
\cite[Sec.~8.2, Steps~5--6]{king2024gapped}. Thus the vertices of the
weighted base graph \(\GH\) split as
\[
  V(\GH)=V_{\mathrm{reg}}\sqcup V_{\mathrm{gad}},
\]
where every register vertex has amplitude weight \(1\) and every added gadget
vertex has amplitude weight \(\lambda\). Equivalently, the binary level
function
\[
  \level(v)=
  \begin{cases}
    0,&v\in V_{\mathrm{reg}},\\
    1,&v\in V_{\mathrm{gad}}
  \end{cases}
\]
satisfies \(w(v)=\lambda^{\level(v)}\) exactly and has maximum level \(L=1\).
Consequently, a simplex \(\sigma\) has weight
\[
  w(\sigma)=\lambda^{|\sigma\cap V_{\mathrm{gad}}|}.
\]

\begin{definition}[Binary expansion correspondence]
\label{def:expansion-correspondence}
Let \(\GammaH\) be the weighted base complex with the binary level
function above. Choose an integer ratio
\(\ratio\ge2\), set \(\lambda^2=1/\ratio\), and take the top block size
\(M=\ratio\). Since \(L=1\), \(M\) is an integer multiple of \(\ratio^L=\ratio\).
Thus register vertices have block size
\(\ratio\), while gadget vertices have block size \(1\). The expanded unweighted
gadget is obtained by the following correspondence:
\[
\begin{array}{c|c}
\text{weighted base object} & \text{expanded unweighted object}\\
\hline
\text{register vertex }v & \text{block }K_{\ratio}\\
\text{gadget vertex }v & \text{singleton }K_1\\
\text{edge }uv & \text{complete bipartite graph }K_{\bfmult_u,\bfmult_v}\\
\text{simplex }\sigma & \text{complete multipartite family over its blocks}\\
\text{amplitude weight }w(v) & \sqrt{\bfmult_v/\ratio}\in\{1,\lambda\}\\
\text{gadget attachment to }v & \text{attachment to the full block }K_{\bfmult_v}
\end{array}
\]
\end{definition}

The last line records the convention used throughout the blow-up. An attachment
to a base vertex is replaced by attachments to all copies in the corresponding
block.

The graph construction is exactly the level-ratio blow-up of
\Cref{def:level-ratio-blowup}; hence the output is a clique complex by
definition. The symmetric-sector identification is \Cref{lem:symmetric-reduction}
applied to the levelled base complex \(\GammaH\). In the present binary case
\(\bfmult_v=M\ratio^{-\level(v)}\) with \(M=\ratio\), so the symmetric-sector
weight from \Cref{lem:symmetric-reduction} is
\[
  \sqrt{\bfmult_v}
  =
  \sqrt M\,\lambda^{\level(v)},
  \qquad
  \lambda^2=1/\ratio.
\]
The factor \(\lambda^{\level(v)}\) is exactly the source vertex weight, and
\Cref{lem:symmetric-reduction} therefore gives the common Laplacian rescaling
by \(M\) on the fully symmetric sector.

In this binary application, nonsymmetric directions can only come from register
blocks, since gadget vertices are not blown up. Hence the decoupling bound is
at least \(\ratio I\) on the nonsymmetric sector. We will only use the weaker
bound \(\bDelta\succeq I\).

\begin{figure}[t]
  \centering
  \includegraphics[width=.66\textwidth]{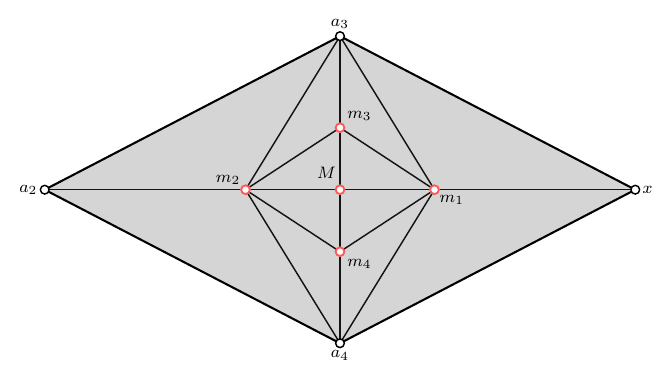}
  \caption{A schematic single-projector filling gadget. The boundary cycle is
  filled by an interior membrane, so the selected encoded cycle becomes
  null-homologous while the orthogonal cycles are preserved by the selectivity
  argument. The drawing suppresses the cut-and-reglue triangulation and
  clique-closure data used in that argument.}
  \label{fig:single-projector-gadget}
\end{figure}

\begin{lemma}[Weight-free filling and selectivity]
\label{lem:weight-free-selectivity}
For every projector \(\phi_i\) from the fixed finite family of
\(\Gtwo\)-integer local projectors recalled in
\Cref{sec:complexity-conventions}, the attached gadget makes the intended
encoded cycle null-homologous. The surviving target-degree register quotient
is exactly \(\ker\phi_i\); in particular, the attachment creates no additional
target-degree class.
\end{lemma}

\begin{proof}
\KK{} Lemma~8.4 shows that the cut-and-reglue, thickening, and coning
construction attaches a clique complex in which the designated cycle is a
boundary~\cite[Lem.~8.4]{king2024gapped}. For the concrete local states used
here, Rudolph verifies that the projected state becomes a boundary, and that the
only remaining target-degree classes are represented by its orthogonal
complement.
\cite[App.~D.1]{rudolph2025universala}. Tensor-product local states are handled
by joins, and padding by unused qubits tensors this quotient with the identity.
Thus the surviving encoded register space is precisely \(\ker\phi_i\). These
are combinatorial-topological statements; no vertex-weight estimate enters.
\end{proof}

The same weight-free statement also gives the YES-case direction. If the source
Hamiltonian has a zero-energy state, the corresponding encoded harmonic vector
extends across all satisfied gadgets. Thus low energy is witnessed
topologically, rather than by a cancellation among weighted matrix entries; the
blow-up construction only has to preserve this encoded harmonic sector.

\section{Energy Analysis of the Expanded Gadget Complex}
\label{sec:energy-analysis}

The NO-case estimate has two logically separate layers. First, a single
projector gadget must have the intended kernel and an isolated penalty band in
the projected-state direction. Second, after all gadgets are attached, these
spectral subspaces must survive the fact that the gadgets share the same qubit
register and therefore share part of the down-Laplacian.
The role of the blow-up is to identify the symmetric sector with the
corresponding weighted source complex, while
\Cref{thm:asymmetric-decoupling} keeps its orthogonal complement above the
relevant energy scale.

\subsection{Single-gadget energy}

Let \(m_i=|\operatorname{supp}(\phi_i)|\le\mloc\) be the number of logical
qubits on which the \(i\)th local projector acts. Let
\(\Gamma_i^{\mathrm{loc}}\) be the weighted single-gadget complex on those
\(m_i\) qubits from the cited construction, and define its identity padding by
\[
  \Gamma_i^{\mathrm{pad}}
  =
  \Gamma_i^{\mathrm{loc}}*\Cl(\mathcal G_{n-m_i}).
\]
Let \(\Delta_i^{\mathrm{pad}}\) be its Hodge
Laplacian in degree \(2n-1\). We identify
\(\mathcal H_r=(\C^2)^{\otimes r}\) with the harmonic \(r\)-qubit register and
write
\[
  \Phi_i=\phi_i\otimes I_{\mathcal H_{n-m_i}},
  \qquad
  \Phi_i^\perp=I_{\mathcal H_n}-\Phi_i
\]
for the ideal projectors on the harmonic register, extended by zero on its
orthogonal complement in the padded chain space. Put
\[
  \eta_i=\lambda^{4m_i+2},
  \qquad
  \eta_{\mathrm{sg}}=\lambda^{4\mloc+2}.
\]
The decomposition below separates the kernel from three positive-energy
ranges. We use two convenient names: the band close to the projected-state
subspace is the penalty band, and the next low-energy band, supported near the
central filling vertex of the gadget, is the bulk band.

\begin{lemma}[Single-gadget spectrum
{\cite[Lems.~9.1 and~10.1]{king2024gapped}}]
\label{lem:single-gadget-spectrum}
There is a constant \(\lambda_0>0\), depending only on the fixed local gadget
family, such that for \(0<\lambda\le\lambda_0\) the padded chain space has an
orthogonal, \(\Delta_i^{\mathrm{pad}}\)-invariant decomposition
\[
  C_{2n-1}(\Gamma_i^{\mathrm{pad}})
  =
  \mathcal K_i\oplus\mathcal F_i\oplus\mathcal B_i\oplus\mathcal A_i
\]
with the following properties.
\begin{enumerate}[label=(\roman*),leftmargin=*]
  \item \(\mathcal K_i=\ker\Delta_i^{\mathrm{pad}}\) has dimension
  \((2^{m_i}-1)2^{n-m_i}\) and is an \(O(\lambda)\)-perturbation of
  \(\im\Phi_i^\perp\).
  \item The penalty band \(\mathcal F_i\) has dimension \(2^{n-m_i}\), is an
  \(O(\lambda)\)-perturbation of \(\im\Phi_i\), and its spectrum is contained
  in
  \[
    [a_{\mathrm{sg}}\eta_i,b_{\mathrm{sg}}\eta_i]
  \]
  for constants \(0<a_{\mathrm{sg}}\le b_{\mathrm{sg}}\) independent of \(i\).
  \item The bulk band \(\mathcal B_i\) is an \(O(\lambda)\)-perturbation of the
  subspace generated by the central-vertex simplices, and its spectrum lies in
  \([a_{\mathrm{bulk}}\lambda^2,b_{\mathrm{bulk}}\lambda^2]\).
  \item The remaining band \(\mathcal A_i\) has spectrum in
  \([a_{\mathrm{high}},b_{\mathrm{high}}]\).
\end{enumerate}
All displayed constants are positive and depend only on the fixed local gadget
family. If \(\widehat\Phi_i^\perp\) and \(\widehat\Phi_i\) denote the orthogonal
projectors onto \(\mathcal K_i\) and \(\mathcal F_i\), respectively, then
\[
  \|\widehat\Phi_i^\perp-\Phi_i^\perp\|=O(\lambda),
  \qquad
  \|\widehat\Phi_i-\Phi_i\|=O(\lambda).
\]
Moreover, since \(m_i\le\mloc\), the penalty-band lower endpoint satisfies
\[
  a_{\mathrm{sg}}\eta_i
  \ge
  a_{\mathrm{sg}}\eta_{\mathrm{sg}}.
\]
\end{lemma}

\begin{proof}[Proof sketch]
For completeness, we recall the \KK{} single-gadget analysis in the form used
here. For the unpadded \(m_i\)-qubit gadget, statements (i)--(iv) are
\cite[Lem.~9.1]{king2024gapped}. The unused register is added by the join
\(\Gamma_i^{\mathrm{loc}}*\Cl(\mathcal G_{n-m_i})\). The augmented join formula
makes its Laplacian a tensor sum; restricting the unused-register factor to
its harmonic space therefore tensors each local low-energy band with
\(\mathcal H_{n-m_i}\), while every nonharmonic unused-register factor has
constant energy. All tensor summands with a nonharmonic unused-register factor
are included in \(\mathcal A_i\), together with the high local band tensored
with \(\mathcal H_{n-m_i}\). The four displayed sectors are therefore
exhaustive. This gives the dimensions and spectral scales above, exactly as in
\cite[Lem.~10.1]{king2024gapped}.

The projector estimates follow from the definition of an
\(O(\lambda)\)-perturbation of subspaces and
\cite[Lem.~7.4]{king2024gapped}. The constants can be chosen uniformly because
Rudolph's source uses a fixed finite gadget family with
\(m_i\le4\)~\cite[Thm.~6.8 and App.~D.1]{rudolph2025universala}. Finally,
\(\eta_i\ge\eta_{\mathrm{sg}}\) follows from \(0<\lambda<1\) and
\(m_i\le\mloc\).
\end{proof}

The many-gadget argument uses these spectral projectors directly. On the
expanded unweighted gadget, \Cref{lem:symmetric-reduction} multiplies all
symmetric-sector bands by \(M\), while \Cref{thm:asymmetric-decoupling} keeps
the orthogonal complement above \(\min_v \bfmult_v\).

\subsection{Many-gadget energy}

Let \(\Delta_{\GammaH,2n-1}\) be the degree-\((2n-1)\) weighted source
Laplacian after all padded gadgets have been attached. The up-Laplacian
contributions from different gadget interiors are additive because the
interiors share no top-dimensional simplices. The down-Laplacian requires a
joint estimate because every gadget meets the same qubit register, so
boundaries of different gadget chains can interfere on the shared register.
This is precisely the interference controlled in the cited many-gadget
argument.

\begin{theorem}[Many-gadget energy bound
{\cite[Thm.~10.1]{king2024gapped}}]
\label{thm:kk-many-gadget}
There are a constant \(\alpha_{\mathrm{wt}}>0\) and a sufficiently large fixed
integer \(C_{\mathrm{wt}}\), depending only on the local gadget family, such
that the following holds. Suppose \(H\) is an \(n\)-qubit source Hamiltonian and
\(H=\sum_{i=1}^t\Phi_i\succeq\pH I_{\mathcal H_n}\), where
\(0<\pH\le1\). Set \(c_{\mathrm{wt}}=C_{\mathrm{wt}}^{-1}\) and
\[
  \lambda=c_{\mathrm{wt}}\pH/t.
\]
Then the weighted combined complex satisfies
\[
  \Delta_{\GammaH,2n-1}
  \succeq
  \alpha_{\mathrm{wt}}\lambda^{4\mloc+2}t^{-1}\pH\,I .
\]
\end{theorem}

For completeness, a proof sketch of this imported estimate in the present
notation is given in \Cref{app:many-gadget-proof-sketch}.

The single- and many-gadget estimates above describe the weighted low-energy
sector transported to the symmetric sector of the unweighted blow-up. The
final NO-case argument now fixes the multiplicity ratio, top block size, and
output-size comparison explicitly.

\subsection{NO-case gap transfer}

\begin{theorem}[NO-case gap for the unweighted blow-up]
\label{thm:unweighted-no-gap}
Assume the source Hamiltonian satisfies
\(\lambda_{\min}(H)\ge \pH\). The blow-up graph \(\bGH\) produced by the
reduction has associated clique complex \(\bGammaH=\Cl(\bGH)\), and it
satisfies
\[
  \lambda_{\min}(\bDelta_{2n-1}(\bGammaH))
  \ge \pgap(N)=N^{-c_{\mathrm{gap}}}
\]
in the NO case, for the fixed constant \(c_{\mathrm{gap}}\) chosen in
the problem definition.
\end{theorem}

\begin{proof}
Combine
\Cref{lem:symmetric-reduction,thm:asymmetric-decoupling,thm:kk-many-gadget}.
The symmetric-sector Laplacian is exactly \(M\) times
the weighted source Laplacian at \(\lambda^2=1/\ratio\), including the
cross-gadget down-Laplacian terms. Recall that \(\pH=1/q(n)\), with \(q(n)\)
a positive integer-valued polynomial. Choose the integer ratio
\[
  \ratio
  =
  \bigl(C_{\mathrm{wt}}tq(n)\bigr)^2.
\]
Here \(C_{\mathrm{wt}}\), \(t\), and \(q(n)\) are positive integers, so
\(\ratio\) is an integer (and \(\ratio\ge2\) after fixing
\(C_{\mathrm{wt}}\) sufficiently large).
Then
\[
  \lambda=\ratio^{-1/2}
  =\frac{1}{C_{\mathrm{wt}}tq(n)}
  =c_{\mathrm{wt}}\pH/t,
\]
so the parameter choice in \Cref{thm:kk-many-gadget} holds exactly, with no
rounding of the perturbative scale.

The weighted base complex has only two vertex levels. Set
\[
  M=\ratio,
  \qquad
  \bfmult_v=
  \begin{cases}
    \ratio,&v\in V_{\mathrm{reg}},\\
    1,&v\in V_{\mathrm{gad}}.
  \end{cases}
\]
Since \(L=1\) and \(M=\ratio\), \(M\) is an integer multiple of \(\ratio^L=\ratio\).
Hence every block size is an integer and \(\min_v \bfmult_v=1\).
If \(R=|V(\GammaH)|\) is the number of vertices in the base complex, the output
vertex count is
\[
  N=\sum_{v\in V(\GammaH)}\bfmult_v
  \le R\ratio.
\]
The source reduction has \(R=\poly(n)\), \(t=\poly(n)\), and
\(\pH\ge1/\poly(n)\), while \(\mloc\) is fixed. Hence \(\ratio\), \(M\), and \(N\) are all
polynomially bounded in \(n\).

On the symmetric sector, the weighted many-gadget energy bound and the exact factor \(M\) give
\[
  E_{\mathrm{sym}}
  \ge
  \ratio\alpha_{\mathrm{wt}}\lambda^{4\mloc+2}t^{-1}\pH
  =
  \alpha_{\mathrm{wt}}\ratio^{-2\mloc}t^{-1}\pH.
\]
Thus there are fixed constants \(c_0>0\) and \(a,d>0\) such that
\[
  E_{\mathrm{sym}}\ge c_0n^{-a},
  \qquad
  n\le N\le c_0^{-1}n^d.
\]
Choose a fixed \(c_{\mathrm{gap}}>a\). For all sufficiently large \(n\),
\[
  N^{-c_{\mathrm{gap}}}
  \le n^{-c_{\mathrm{gap}}}
  \le c_0n^{-a}
  \le E_{\mathrm{sym}}.
\]
Increasing the fixed target exponent \(c_{\mathrm{gap}}\), if necessary,
covers the finitely many smaller reduced instances.
The orthogonal complement of the symmetric sector has gap at least
\(\min_v \bfmult_v=1\), so taking the minimum of the two bounds proves the theorem.
\end{proof}

\section{Containment}
\label{sec:containment}

\begin{proposition}[Containment]
\label{thm:containment}
Unweighted Gapped Clique Homology belongs to \(\QMAone^{\Gtwo}\).
\end{proposition}

\begin{proof}
Let the input graph \(G\) have \(N\) vertices. The verifier uses one qubit for
each vertex, so the ambient Hilbert space is \((\C^2)^{\otimes N}\).
Encode a simplex by its increasing vertex set. Thus a computational-basis
string \(\ket{S}\) is valid when \(S\subseteq V(G)\) has cardinality \(k+1\)
and induces a clique; the increasing order fixes its orientation. Let
\(\mathcal V_k\) be the span of these valid strings. Extend the
clique Hodge Laplacian to the full computational-basis space by
\[
  \widetilde\Delta_k
  =
  \Delta_k(\Cl(G))\big|_{\mathcal V_k}
  \oplus I_{\mathcal V_k^\perp}.
\]
The invalid strings therefore have energy one and cannot create additional
zero modes.

The matrix entries of \(\widetilde\Delta_k\) are integers arising from
incidence signs and the invalid-string penalty. For a fixed valid basis state
\(\ket{\sigma}\), the possible nonzero entries
\(\bra{\tau}\widetilde\Delta_k\ket{\sigma}\) are obtained by the two incidence
walks
\[
  \sigma
  \xrightarrow{d_k}
  \eta
  \xrightarrow{\partial_{k+1}}
  \tau
  \qquad
  \text{or}
  \qquad
  \sigma
  \xrightarrow{\partial_k}
  \mu
  \xrightarrow{d_{k-1}}
  \tau .
\]
There are at most \(N(N+1)\) candidates of each type, and duplicate
destinations are combined with their incidence signs. Thus
\(\widetilde\Delta_k\) has sparsity and absolute entry sum per basis state at
most
\[
  r(N)=2N(N+1)+1.
\]
On an invalid string, the only nonzero entry is the diagonal penalty. Clique
validity, path enumeration, aggregation of duplicate destinations, and exact
entry evaluation use polynomially many queries to the adjacency description
of \(G\).

Choose the power of two
\[
  S=2^{\lceil\log_2 r(N)\rceil}.
\]
Then
\[
  H_{G,k}=S^{-1}\widetilde\Delta_k
\]
is an exactly specified sparse Hamiltonian with \(\|H_{G,k}\|\le1\) and entries
in \(\mathbb Q\subset\mathbb Q(i)\). In the YES case its smallest singular value is exactly
zero. In the NO case, positivity of the Hodge Laplacian and
\(\pgap(N)\le1\) give
\[
  \sigma_1(H_{G,k})
  =\lambda_{\min}(H_{G,k})
  \ge \frac{\pgap(N)}{S}
  =\frac{1}{\poly(N)}.
\]

Rudolph calls this the Exact Sparse Hamiltonian problem
\cite[Prob.~2.9]{rudolph2025universala}. For entries in \(\mathbb Q(i)\), his
containment theorem places it in
\(\QMAone^{\Gtwo}\)~\cite[Lem.~6.1]{rudolph2025universala}. The exact LCU
decomposition into implementable one-sparse unitaries is supplied by
\cite[Lem.~6.2]{rudolph2025universala}. Applying that verifier to \(H_{G,k}\)
gives perfect completeness for a harmonic witness and inverse-polynomial
soundness in the promised NO case. 
\end{proof}

\section{Discussion}
\label{sec:discussion}

The theorem separates two roles that are intertwined in the weighted
construction. Vertex weights provide a convenient metric for the gap analysis,
but they are not the source of the computational hardness. The source is the
homological encoding of a local Hamiltonian together with an inverse-polynomial
spectral promise. The blow-up construction makes this separation explicit. It
recovers the weighted Hodge metric on the symmetric sector, and the decoupling
estimate removes the copy-dependent directions from the low-energy space.

This viewpoint raises the question of whether similar unweighting mechanisms
apply to other weighted or filtered homological complexity results. The present
paper only treats the weighted clique-homology reduction needed for the
\(\QMAone\) hardness result
\cite{king2024gapped,rudolph2025universala}. Related settings include
weighted-complex \(\mathsf{QMA}\)-hardness for Hodge Laplacians
\cite{rayudu2025fermionic}, \(\mathsf{BQP}_1\)-hard harmonic persistence
\cite{gyurik2024quantum}, \(\mathsf{MA}\)-complete orientable-filtration
homology~\cite{hayakawa2025computationalMA}, and \(\mathsf{DQC1}\)-hard
normalized persistence \cite{lowe2026complexity}. Understanding which weights
or filtration values can be replaced by unweighted multiplicity would help
clarify the role of artificial weighting in worst-case quantum
advantages for topological data analysis.

\paragraph{Sparsifying the blow-up.}
The present reduction is intentionally dense: each base vertex is replaced by a
complete graph, and each base edge by a complete bipartite connection. A natural
next question is whether the same multiplicity-as-weight mechanism can be
implemented by sparser replacement graphs, for example by expanders or other
pseudorandom block gadgets. Such a sparsification would be important for eventual
homological QPCP-type questions.

The main obstruction in this direction is topological rather than merely spectral. In the dense
blow-up, complete multipartite structure and block-relabeling symmetry make
the intended sector rigid, while \Cref{thm:asymmetric-decoupling} pushes its
orthogonal complement away from zero. If the block connections are sparsified, this symmetry
is lost or weakened, and new unintended homologous cycles may appear. These extra holes
could overwhelm the intended encoded homology even if the graph remains a good
spectral approximator in the usual graph-Laplacian sense. Thus sparsifying the
construction is not simply a matter of preserving spectrums but it must also
control spurious homology.

One possible way to formulate a more tractable problem is to mark the intended
cycles and study a persistence-style promise. Instead of asking for the full
homology of the sparse complex to match the dense blow-up, one could ask how long the specified encoded cycles persist across a controlled
filtration.

\section*{Acknowledgments}
RH thanks Kazuki Sakamoto for helpful discussion.
RH was supported by JST PRESTO Grant Number JPMJPR23F9 and JST ASPIRE Grant Number JPMJAP26A4, Japan.

\section*{AI Disclosure}
The author used generative AI tools, including Anthropic Claude Opus 4.7 and OpenAI GPT-5.5/GPT-5.6/Codex, for brainstorming, outlining, copy-editing,
LaTeX assistance, and notation checks. The author verified all mathematical content, citations, and final text, and takes full responsibility for the manuscript.

\appendix
\section{Proof sketch for the many-gadget estimate}
\label{app:many-gadget-proof-sketch}

For completeness, this appendix gives a proof sketch of the imported
many-gadget estimate from \cite[Thm.~10.1 and App.~B]{king2024gapped} in the
notation used in \Cref{sec:energy-analysis}.

\begin{proof}[Proof sketch of \Cref{thm:kk-many-gadget}]
In the notation of \cite{king2024gapped}, the source promise \(g\) is our
\(\pH\), the number of Hamiltonian terms is \(t\), and the local width \(m\)
is bounded by our fixed constant \(\mloc\).
Their theorem chooses \(\lambda=c\,t^{-1}g\) for a sufficiently small constant
\(c\) and obtains the scale
\(E=c\lambda^{4m+2}t^{-1}g\). We choose such a constant of the form
\(c_{\mathrm{wt}}=C_{\mathrm{wt}}^{-1}\), and absorb its fixed prefactor into
\(\alpha_{\mathrm{wt}}\). Since \(m\le\mloc\) and \(0<\lambda<1\), this bound
may be weakened to the displayed common exponent \(4\mloc+2\). We state the
three imported estimates from their Appendix~B below.

The recalled argument starts from the single-gadget spectral decomposition in
\Cref{lem:single-gadget-spectrum}. Let \(\Pi_0\) project onto the qubit-register
chain space, let \(\Pi_i\) project onto the chains containing interior vertices
of gadget \(i\), and let \(\Pi_{\mathcal H}\le\Pi_0\) project onto the harmonic
register \(\mathcal H_n\). All these projectors are in degree \(2n-1\).
Interior disjointness gives
\[
  \Pi_0+\sum_i\Pi_i=I
\]
on the combined chain space. For each individual padded gadget, the four
spectral projectors from \Cref{lem:single-gadget-spectrum} decompose the
two-piece chain space \(\im(\Pi_0+\Pi_i)\). Thus
\[
  \Pi_0+\Pi_i
  =
  \Pi_i^{(\mathcal A)}+\Pi_i^{(\mathcal B)}
  +\widehat\Phi_i+\widehat\Phi_i^\perp.
\]
Here \(\Pi_i^{(\mathcal A)}\) and \(\Pi_i^{(\mathcal B)}\) project onto the
remaining high-energy band and the bulk band, while \(\widehat\Phi_i\) and
\(\widehat\Phi_i^\perp\) project onto the penalty band and kernel. This is the
flat decomposition used over all \(t\) gadgets.

Let \(\varphi\) be normalized and suppose, toward a contradiction, that
\[
  \langle\varphi,\Delta_{\GammaH,2n-1}\varphi\rangle<E.
\]
For comparison with the source proof, write
\(\varphi=\omega_0+\sum_i\omega_i\), where
\(\omega_0=\Pi_0\varphi\), \(\omega_i=\Pi_i\varphi\), and set
\(\varphi_i=(\Pi_0+\Pi_i)\varphi=\omega_0+\omega_i\).
The three estimates
\cite[Lems.~10.2--10.4]{king2024gapped} give
\begin{align}
  \left\langle\varphi,
    \sum_i\widehat\Phi_i\varphi\right\rangle
  &=
  O\!\left(\lambda^{-(4\mloc+2)}Et\right),
  \label{eq:kk-penalty-overlap}\\
  \left\langle\varphi,
    \sum_i\Pi_i^{(\mathcal A)}\varphi\right\rangle
  &=O(\lambda^2t),
  \label{eq:kk-high-overlap}\\
  \left\langle\varphi,
    \sum_i\Pi_i^{(\mathcal B)}\varphi\right\rangle
  &=O(\lambda^2t).
  \label{eq:kk-bulk-overlap}
\end{align}
The second estimate applies when \(E\) is below a fixed multiple of
\(\lambda^2\), and the third when it is below a fixed multiple of
\(\lambda^4\); the value of \(E\) chosen below satisfies both conditions.

\Cref{eq:kk-penalty-overlap} uses only the up-Laplacian:
the gadget up-Laplacians add without interference and dominate the penalty
projectors at scale \(\lambda^{4m_i+2}\). For
\Cref{eq:kk-high-overlap}, restricting to \(\varphi_i\) leaves a possible
down-boundary contribution on the shared register; the weighted gadget
boundaries make this leakage \(O(\lambda)\), so the local energy is
\(O(\lambda^2)\). For \Cref{eq:kk-bulk-overlap}, the bulk localization claim
in the same proof places the relevant boundary or coboundary inside the
pairwise-disjoint bulk regions. Consequently those contributions cannot cancel
between gadgets. These are exactly the two places where the shared
down-Laplacian needs more than up-Laplacian additivity.

It remains to insert these estimates into the flat decomposition. As operators
on the full chain space,
\[
  \Phi_i^\perp=\Pi_{\mathcal H}-\Phi_i,
  \qquad
  \sum_i\Phi_i=H\quad\text{on }\mathcal H_n.
\]
The projector perturbation estimate in
\Cref{lem:single-gadget-spectrum} therefore gives
\begin{align*}
  \sum_i\widehat\Phi_i^\perp-(t-1)\Pi_0
  &=
  \Pi_{\mathcal H}-H
  -(t-1)(\Pi_0-\Pi_{\mathcal H})
  +O(\lambda t).
\end{align*}
Here the final \(O(\lambda t)\) denotes an operator whose norm has that order;
all implicit constants depend only on the fixed local gadget family.
Using the two projector identities above and then
\(H\succeq\pH\Pi_{\mathcal H}\), we obtain
\begin{align*}
  1
  &=
  \left\langle\varphi,
    \left[
      \sum_i\bigl(
        \widehat\Phi_i^\perp+\widehat\Phi_i
        +\Pi_i^{(\mathcal A)}+\Pi_i^{(\mathcal B)}
      \bigr)
      -(t-1)\Pi_0
    \right]\varphi
  \right\rangle\\
  &\le
  (1-\pH)\|\Pi_{\mathcal H}\varphi\|^2
  -(t-1)
    \langle\varphi,(\Pi_0-\Pi_{\mathcal H})\varphi\rangle\\
  &\qquad
  +O(\lambda t)
  +O\!\left(\lambda^{-(4\mloc+2)}Et\right)
  +O(\lambda^2t).
\end{align*}
Set
\[
  \lambda=c_{\mathrm{wt}}\pH/t,
  \qquad
  E=\alpha_{\mathrm{wt}}
    \lambda^{4\mloc+2}t^{-1}\pH.
\]
The three error terms are respectively
\(O(c_{\mathrm{wt}}\pH)\), \(O(\alpha_{\mathrm{wt}}\pH)\), and
\(O(c_{\mathrm{wt}}^2\pH)\), where the last estimate uses
\(\pH\le1\) and \(t\ge1\). The second term on the right-hand side is
nonpositive. Choosing first \(c_{\mathrm{wt}}\), and then
\(\alpha_{\mathrm{wt}}\), sufficiently small makes the total error strictly
less than \(\pH\). Since
\(\|\Pi_{\mathcal H}\varphi\|\le1\), the displayed inequality would then give
\(1<1\), a contradiction. Hence
\[
  \Delta_{\GammaH,2n-1}
  \succeq
  \alpha_{\mathrm{wt}}\lambda^{4\mloc+2}t^{-1}\pH\,I .
\]
\end{proof}

\bibliographystyle{alpha}
\bibliography{references}

\end{document}